\newif\ifjournalformat 
\journalformattrue  

\ifjournalformat
    \documentclass[journal]{IEEEtran}
    \newcommand{\tikzmedium}{\columnwidth}
    \newcommand{\tikzlarge}{0.5\textwidth}
    \newenvironment{widefigure}{\begin{figure*}[t]}{\end{figure*}}
\else
    \documentclass[12pt, draftclsnofoot, onecolumn]{IEEEtran}
    \newcommand{\tikzmedium}{0.5\textwidth}
    \newcommand{\tikzlarge}{0.7\textwidth}
    \newenvironment{widefigure}{\begin{figure}[t]}{\end{figure}}
\fi

\usepackage{cite}
\usepackage{amsmath,graphicx}
\usepackage{amssymb}
\usepackage{amsfonts}
\usepackage{graphicx}
\usepackage[english]{babel}
\usepackage{color,latexsym,hyperref,amsmath,cite,amsthm}
\usepackage[mathscr]{euscript}
\usepackage{tikz}
\usetikzlibrary{%
    decorations.pathreplacing,%
    decorations.pathmorphing%
}
\usetikzlibrary{math} 

\tikzset{
        interface_right/.style={        
        postaction={draw,decorate,decoration={border,angle=-45,
                    amplitude=0.1cm,segment length=2mm}}},
                    interface_left/.style={        
        postaction={draw,decorate,decoration={border,angle=45,
                    amplitude=0.1cm,segment length=2mm}}},
}

\usepackage{subcaption}  

\newcommand{\Prob}{\mathbf{P}}
\renewcommand{\P}{\mathbf{P}}

\newcommand{\R}{\mathbb{R}}

\newtheorem{Theorem}{Theorem}[section]

\newtheorem{Proposition}[Theorem]{Proposition}
\theoremstyle{definition}

\newtheorem{Remark}[Theorem]{Remark}

\begin{document}

\title{Reflected wireless signals under random spatial sampling}

\author{H. Paul Keeler
\thanks{H.P. Keeler is with the School of Mathematics and Statistics, University of Melbourne, Australia. (e-mail: hpkeeler@unimelb.edu.au) }
}

\maketitle

\begin{abstract}
We present a propagation model showing that a transmitter randomly positioned in space generates unbounded peaks in the histogram of the resulting power, provided the signal strength is an oscillating or non-monotonic function of distance. Specifically, these peaks are singularities in the empirical probability density that occur at turning point values of the deterministic propagation model. We explain the underlying mechanism of this phenomenon through a concise mathematical argument. This observation has direct implications for estimating random propagation effects such as fading, particularly when reflections off walls are involved. 

Motivated by understanding intelligent surfaces, we apply this fundamental result to a physical model consisting of a single transmitter between two parallel passive walls. We analyze signal fading due to reflections and observe power oscillations resulting from wall reflections---a phenomenon long studied in waveguides but relatively unexplored in wireless networks. For the special case where the transmitter is placed halfway between the walls, we present a compact closed-form expression for the received signal involving the Lerch transcendent function. The insights from this work can inform design decisions for intelligent surfaces deployed in cities.
\end{abstract}

\begin{IEEEkeywords}
Distribution, Lerch function, reflection, received signal power, random phases.
\end{IEEEkeywords}

\section{Introduction}\label{sec:Introduction}
Imagine a single wireless transmitter and an observer located some fixed distance apart. Assume the  signal power is a  function of distance. Provided that the function has  turning points, then if we repeatedly place the transmitter at random, the resulting randomized signal power has a probability density with peaks or, more specifically, singularities. These singularities create technical challenges for mathematical and statistical treatments of wireless networks, as essentially all mathematical models and methods assume no such singularities. In this paper we show that the conditions for the singularities, which are asymptotically the inverse square root kind, naturally arise in a model of wireless signals reflecting off building walls.

When electromagnetic signals collide with objects, they undergo reflection and absorption and may take multiple paths. A popular approach is to model these phenomena using ray tracing, which is a deterministic method based on classical results from physics, more specifically optics, such as the Snell--Descartes law. This technique is, in turn, derived from Fermat's principle, which states that electromagnetic waves travel paths that take the least amount of time (more precisely, are stationary with time). These results allow researchers to estimate signal strengths by tracing signal paths as they collide with obstacles.

To incorporate randomness into signal models, researchers have also developed models that provide a statistical description of how such collisions cause signal strengths to vary. Understanding signal propagation in cities is important for emerging technologies such as intelligent reconfigurable surfaces and  deployments of fifth and sixth generation networks. A relevant technological consideration is intelligent (reconfigurable) surfaces. These smart or intelligent walls grant some control over the phases of reflected signals, with the aim of removing the effects of destructive interference~\cite{di2019smart}.

Understanding how random transmitter positioning creates probability density singularities is particularly important for deployment of intelligent surfaces. While this technology aims to control reflected signal phases, the fundamental geometric relationships that create turning points---and thus singularities---persist. Optimally tuning intelligent surfaces requires understanding where these singularities occur and how phase control can mitigate their effects on overall system performance.

To this end, we develop a propagation model---or, more precisely, a fading model---for describing electromagnetic wave propagation between two parallel walls by employing the method of images. The method of images is a standard technique often used in, for example, electrostatics to model electric fields induced by charges near conducting surfaces~\cite[Chapter~2]{jackson2012classical},~\cite[Section~3.2]{griffiths2018introduction}, appearing long ago in the classic work by Maxwell~\cite[Chapter~XI]{maxwell1881treatise}. It is a general method for incorporating boundary conditions into models of physical systems.

We examine the resulting models using mathematical techniques, which we then support with numerical and simulation results.  

\subsection{Contributions}
The paper makes the following contributions:
\begin{itemize}
\item We establish that randomly positioning a transmitter produces singularities in the probability density of received signal power at values corresponding to turning points of the deterministic power function.

\item For the symmetric two-wall configuration, we derive a closed-form expression for the received signal using the Lerch transcendent function.

\item We develop a physical-statistical propagation model for parallel walls using the method of images that bridges deterministic ray-tracing and purely statistical fading models.

\item We demonstrate both theoretically and numerically that constructive and destructive wave interference creates turning points in signal power, which manifest as peaks (singularities) in the probability density under random transmitter placement.

\item We show that the locations of these density singularities can be predicted from the geometry of the configuration, providing insight into signal fading in environments with intelligent surfaces.
\end{itemize}

\section{Related work}\label{sec:related.work}

\subsection{Random Propagation Models with Singularities}
There is a wide range of statistical propagation models for fading and shadowing~\cite{simon2005digital}. Although singularities can appear in density expressions, they are usually edge cases existing at the origin, as with the Nakagami and log-normal models, which are otherwise smooth models. Apart from that, we are unaware of random propagation models with singularities in their probability densities. We speculate that this is because such models are often developed using specific approaches. Given some independence assumptions, applying a suitable transformation and the central limit theorem produces a model based on normal (or Gaussian) random variables. Conversely, we discovered these densities by first developing a deterministic model (with an oscillating component) and then randomizing it.

Another reason is that it is considerably more complicated to work with such densities. There is much literature on statistically estimating probability densities, but considerably less on densities with discontinuities such as singularities. Generally speaking, much of statistical density estimation theory hinges upon smoothness assumptions on the underlying density functions; see Remark~\ref{rem:density_est} for details and references.

\subsection{Electromagnetic Signals in Streets}
To understand how electromagnetic signals propagate down a city street, Amitay~\cite{amitay1992modeling} developed in the early 1990s what he called the canyon model, consisting of two infinitely long parallel walls. The work is based on computer simulations that consider in total ten reflected signals or rays, which include both ground and wall reflections, whereas we only examine wall reflections. Amitay included only the first and second images in his model.

There are decades of work on how electromagnetic waves travel in waveguides, which are usually linear structures designed for such waves to propagate with little power attenuation. Researchers have also applied this theory to the problem of fading due to buildings. For example, independently of our work, Kyritsi studied in her doctorate thesis a mathematical model~\cite[Section~7.1.3]{kyritsi2001multiple} that is essentially the same as our non-random one when the walls are equidistant. Motivated by understanding multiple-input and multiple-output systems, she also used the method of images to derive the same expression as ours with an infinite series representing an infinite number of images. She then truncated to a finite number of terms, but no connection was made to the Lerch function, which we employ in our results. Kyritsi~\cite[Section~7.1.4]{kyritsi2001multiple} also used the method of images to study a rectangular model, meaning a room with four walls. Cox and Kyritsi~\cite{kyritsi2001propagation} developed this model in a paper where they compared it to real-world measurements. No randomness was added to these models, so beyond the initial non-random model setup, their work does not overlap with ours.

A more rigorous treatment of electromagnetic reflections is  \emph{exact image theory} developed by  Lindell and Alanen~\cite{lindell1984exact_part1,lindell1984exact_part2,lindell1984exact_part3},  which accounts for polarization-dependent reflection coefficients, angle-dependent effects, and near-field corrections that the classical method of images neglects. Although exact image theory provides greater physical accuracy, for demonstrating our key results, the classical method of images suffices.

Researchers have also examined signal propagation in tunnels, starting in the 1970s, such as the paper by Emslie, Lagace, and Strong~\cite{emslie1975theory}, who used rectangular waveguide models. Some work has used the so-called uniform (geometrical) theory of diffraction, which combines ideas from Maxwell's theory of electromagnetism and classical geometrical optics, such as Fermat's principle~\cite{mcnamara1990introduction}. However, we have found no research on tunnel propagation that is sufficiently relevant to our work; see the survey on tunnel propagation~\cite{hrovat2014survey}.

Beyond the aforementioned work, there appears to be little that is particularly relevant to our new model and results. The reason for this, we argue, is that most signal propagation models fall into two categories: purely random, such as Rayleigh fading, or purely deterministic, such as ray tracing. Our model is a hybrid of these two approaches, producing potentially new results and observations.

\subsection{Stochastic Geometry Models}
Our work only examines a model with a single transmitter-receiver pair, but it raises obvious questions for the multi-transmitter setting. Here it is natural to employ techniques from point processes and stochastic geometry. This provides a mathematical framework for analyzing wireless networks of randomly positioned nodes, typically with Poisson point processes; see the textbooks by Baccelli and B{\l}aszczyszyn~\cite{BB09} and Haenggi~\cite{haenggi2012stochastic}. In the last couple of years work has already commenced on examining intelligent surfaces using stochastic geometry models~\cite{zhang2021reconfigurable,shafique2022stochastic,lee2025much}, although at the moment we see little overlap with our current work.

The Poisson assumption is partly supported by mathematical work showing that randomly weakening multiple signals can make a non-Poisson wireless network appear Poisson to a single user~\cite{keeler2018wireless,keeler2016stronger}. This is the case even if there is a certain degree of dependence among the random effects weakening the signals~\cite{ross2017wireless}. One potential research direction is sensibly combining these results and our current work, and examining, for example, the effects of the singularities in the density expressions. With sufficient randomness and regular walls, we conjecture that the Poisson approximation and convergence results will still hold. But the use of intelligent surfaces could potentially introduce too much dependence into the model, breaking any Poisson results.




\section{Notation}\label{sec:notation}
Table~\ref{tab:notation} summarizes the key notation used throughout this paper.

\begin{table}[t]
\centering
\caption{Key notation and symbols}
\label{tab:notation}
\begin{tabular}{cl}
\hline
\textbf{Symbol} & \textbf{Description} \\
\hline
\multicolumn{2}{c}{\textit{Parameters}} \\
$j$ & Imaginary unit, $j := \sqrt{-1}$ \\
$\beta$ & Attenuation exponent \\
$\kappa$ & Power reflection coefficient \\
$k$ & Wave number, $k = 2\pi/\lambda$ \\
$\lambda$ & Signal wavelength \\
$a, b$ & Distances from receiver $O$ to right and left walls \\
$d$ & Wall separation distance, $d = a + b$ \\
\hline
\multicolumn{2}{c}{\textit{Spatial Variables}} \\
$O$ & Receiver location at the origin \\
$z$ & Transmitter location \\
$r$ & Distance from transmitter to receiver located at origin $O$ \\
$x, y$ & Cartesian coordinates of transmitter \\
$\theta$ & Polar angle \\
$\hat{r}_n$ & Distance traversed via $n$-th right image \\
$\hat{\ell}_n$ & Distance traversed via $n$-th left image \\
\hline
\multicolumn{2}{c}{\textit{Signal Functions}} \\
$S(r)$ & Received signal from reflections only \\
$\hat{S}(r)$ & Total received signal (including line-of-sight) \\
$P(r)$ & Received power from reflections, $P(r) = |S(r)|^2$ \\
$\hat{P}(r)$ & Total received power, $\hat{P}(r) = |\hat{S}(r)|^2$ \\
$\alpha(r)$ & Attenuation function \\
$w(r)$ & Wave (oscillating) function \\
$\Phi(\zeta,s,\gamma)$ & Lerch transcendent function \\
\hline
 \multicolumn{2}{c}{\textit{Probability}} \\
 $t_i$ & Location of $i$-th turning point \\
 $F_V(v)$ & Cumulative distribution function of random variable $V$ \\
 $f_V(v)$ & Probability density function \\
 \hline
\end{tabular}
\end{table}

\section{Turning Points in the Propagation Model}\label{sec:Turning}
On the plane consider a single transmitter located at a distance $r$ from an observer at the origin~$O$. Assume the received signal at the origin $S(r)$ is a deterministic function of $r$. For the moment, assume the signal $S(r)$ can be  decomposed as
\begin{align}\label{e.signal}
S(r)=A(r)W(r), \quad r\geq 0.
\end{align}
Here $W(r)$ is an oscillating (possibly complex) function, while $A(r)$ is a (real) non-negative function, which we assume is monotonically decreasing in $r$. For $r>0$, we assume that both $A(r)$ and $W(r)$ are continuous and differentiable. In physical terms, the functions $W(r)$ and $A(r)$ respectively model the wave and attenuation aspects of the propagated signal. 

The received power is given by
\begin{align}\label{e.power}
P(r):=|S(r)|^2=A(r)^2|W(r)|^2.
\end{align}
 Note that we may equivalently write $P(x,y)$ when using Cartesian coordinates, where the relationship to polar coordinates is given by $r=\sqrt{x^2+y^2}$. The decomposition $S(r)=A(r)W(r)$ provides conceptual insight: the oscillating function $W(r)$ creates power oscillations while $A(r)$ modulates their amplitude, which then influences the signal power $P(r)$. However, in complex multipath scenarios, as we develop later, the received signal may not decompose cleanly into this form---yet we will see that the same phenomena still arise.

The function $P$ has local minima and maxima or \emph{turning points}, where its derivative $P'$ is zero. (More generally, these are \emph{stationary points}, since saddle points of a multi-variable signal power $P(x,y)$ can also generate singularities.) This function is non-random, but we are interested in the effects of randomly placing the transmitter. We will show that a turning point of the signal power function $P$ generates a division by zero. This, in turn, creates a singularity in the probability density of the randomized version of the signal power, even when it has a continuous probability distribution. 

\subsection{Peaks in the Probability Density}
In a bounded sample interval $(r_{\text{L}},r_{\text{U}})$ for the radial distance, we assume the function $P$ has $m$ turning points denoted by $t_1,\ldots,t_m$, where $r_{\text{L}} \leq t_1,\ldots,t_m \leq r_{\text{U}}$,  meaning $P'(t_i) = 0$ for $i=1,\dots,m$. We assume the function $P$ is strictly increasing or decreasing between any two consecutive turning point locations $t_i$ and $t_{i+1}$, as well either endpoint $r_{\text{L}}$ or $r_{\text{U}}$ and its nearest turning point. We can now examine what happens when we randomly place the transmitter, which gives a randomized version of the signal power. 

\begin{widefigure}
\centering
\resizebox{\tikzlarge}{!}{


\begin{tikzpicture}[scale=1]
    \draw[->, thick] (-0.5,0) -- (10.5,0) node[right] {$r$};
    \draw[->, thick] (0,-1.5) -- (0,1.5) node[above] {$h(r)$};
    
    \def\ra{0}
    \def\rb{1.57}
    \def\rc{4.71}
    \def\rd{7.85}
    \def\re{10.99}  
    
    
    \tikzmath{\rmax = 3*3.141592;} 

    \draw[very thick, blue, domain=0:\rmax, samples=200, smooth] 
        plot (\x, {sin(\x r)}); 
    
    \pgfmathsetmacro{\ta}{1.5708}  
    \pgfmathsetmacro{\tb}{4.7124}  
    \pgfmathsetmacro{\tc}{7.8540}  
    
    \fill[red] (\ta, 1) circle (2pt);
    \fill[red] (\tb, -1) circle (2pt);
    \fill[red] (\tc, 1) circle (2pt);
    
    \draw[dashed, gray] (\ta, 0) -- (\ta, 1);
    \draw[dashed, gray] (\tb, 0) -- (\tb, -1);
    \draw[dashed, gray] (\tc, 0) -- (\tc, 1);
    
    \node[below,right] at (0, -0.3) {$r_{\text{L}}$};
    \node[below] at (\ta, -0.1) {$t_1$};
    \node[above] at (\tb, +0.1) {$t_2$};
    \node[below] at (\tc, -0.1) {$t_3$};
    \node[below] at (\rmax, -0.1) {$r_{\text{U}}$};
    
    
    \draw[decorate, decoration={brace, amplitude=8pt, mirror}] 
        (0,-1.4) -- (\ta,-1.4) node[midway, below=10pt] {$A_1 = (r_{\text{L}}, t_1)$};
    \draw[decorate, decoration={brace, amplitude=8pt, mirror}] 
        (\ta,-1.4) -- (\tb,-1.4) node[midway, below=10pt] {$A_2 = (t_1, t_2)$};
    \draw[decorate, decoration={brace, amplitude=8pt, mirror}] 
        (\tb,-1.4) -- (\tc,-1.4) node[midway, below=10pt] {$A_3 = (t_2, t_3)$};
    \draw[decorate, decoration={brace, amplitude=8pt, mirror}] 
        (\tc,-1.4) -- (\rmax,-1.4) node[midway, below=10pt] {$A_4 = (t_3, r_{\text{U}})$};
    
    \node[blue, font=\normalsize] at (0.5, 0.8) {$g_1$};
    \node[blue, font=\normalsize] at (3.2, -0.7) {$g_2$};
    \node[blue, font=\normalsize] at (6.3, 0.8) {$g_3$};
    \node[blue, font=\normalsize] at (9.1, 0.8) {$g_4$};
    
    \node[blue, font=\large] at (9.5, 1.5) {$h(r) = \sin(r)$};
    
    
\end{tikzpicture}

}
\caption{A wide class of non-monotonic functions, such as $h(r)=\sin(r)$ on an interval $(r_{\text{L}},r_{\text{U}})$, can be decomposed into a collection of strictly monotonic functions $g_1,\dots,g_{\ell}$ by suitably partitioning the function's domain.  $t_i$ are the locations of the turning points where the derivative $h'(t_i) = 0$. Partition: $(r_{\text{L}}, r_{\text{U}}) = A_0\cup A_1 \cup A_2 \cup A_3 \cup A_4$, where the set $A_0$ contains all the endpoints. The partition is such that each function $g_i$, which is defined as $h$ restricted to $A_i$, is strictly monotonic on $A_i$ for $i=1,\dots,4$. }\label{fig:decomp}
\end{widefigure} 

\subsubsection{No Turning Points} 
We consider the open interval $(t_i,t_{i+1})$ such that the function $P$ is strictly monotonic on this interval; refer to Figure~\ref{fig:decomp} for an example. Let the function $g_i$ be the restriction of $P$ on this interval, meaning $g_i(r)=P(r)$ whenever $r\in(t_i,t_{i+1})$, so $g_i$ is a strictly monotonic function with domain $(t_i,t_{i+1})$, which has a well-defined inverse $g_i^{-1}$. The function $g_i$ is continuous, which, coupled with its strict monotonicity on the interval $(t_i,t_{i+1})$, means that $g_i$ is a one-to-one (or bijective) function on the interval $(t_i,t_{i+1})$~\cite[Proposition~9.8.3]{tao2016analysis}. (These assumptions make $g_i$ an ideal function for transforming a random variable.) 

For clarity, we further assume that $g_i$ is strictly increasing on the interval $(t_i,t_{i+1})$, but a similar analysis applies if it is strictly decreasing. Now we can give mathematical expressions for the (cumulative) probability distribution and the probability density of the received signal when the transmitter is randomly positioned.
\begin{Proposition}\label{prop:monotonic}
Let $U$ be a continuous random variable defined on the interval $(t_i,t_{i+1})$ with probability distribution $F_U(u)=\P(U\leq u)$. Assume the above conditions on the function $g_i$, which is strictly increasing on the interval $(t_i,t_{i+1})$. The random variable $V_i=g_i(U)$ has the probability distribution
\begin{align}
F_{V_i}(v) := \P(g_i(U)\leq v) = F_U(g_i^{-1}(v)), \quad v\in(v_i,v_{i+1}),
\end{align}
where the endpoints are $v_i:=g_i(t_i)$ and $v_{i+1}:=g_i(t_{i+1})$. The probability density of $V_i=g_i(U)$ at $v\in(v_i,v_{i+1})$ is given by
\begin{align}\label{e.densitymonotonic}
f_{V_i}(v) := \frac{dF_{V_i}(v)}{dv} = \frac{f_U(g_i^{-1}(v))}{|g_i'(g_i^{-1}(v))|}.
\end{align}
\end{Proposition}

\begin{Remark}
For the probability density $f_{V_i}$, the absolute value in expression~\eqref{e.densitymonotonic} is not needed because the function $g_i$ is increasing. It is only needed for decreasing $g_i$, so expression~\eqref{e.densitymonotonic} holds for any strictly monotonic $g_i$. For the probability distribution $F_{V_i}$, a decreasing $g_i$ requires taking the complement, in addition to reversing the endpoints of the interval, so $v\in(v_{i+1},v_i)$.
\end{Remark}

\begin{proof}
See Appendix~\ref{s.proof.prop:monotonic}.
\end{proof}

The probability density expression~\eqref{e.densitymonotonic} shows that the first derivative appears in the denominator, which becomes zero at turning points in $P$, revealing the source of the singularities. The probability density $f_{V_i}(v)$ is not defined at these singularities, but we can examine how it behaves asymptotically when we include a single turning point.

\begin{figure}[t]
\centering
\resizebox{\tikzmedium}{!}{
\begin{tikzpicture}

\begin{scope}[shift={(-6,0)}]
    \draw[->, thick] (0,0) -- (5,0) node[right] {$r$};
    \draw[->, thick] (0,0) -- (0,4.5) node[above] {$P(r)$};
    
    \draw[blue, very thick, domain=0.5:4.5, samples=100] 
        plot (\x, {0.8 + 0.4*(\x-2.5)^2});
    
    \coordinate (t) at (2.5, 0);
    \coordinate (Pt) at (0, 0.8);
    \coordinate (turnpt) at (2.5, 0.8);
    
    \draw[dashed, gray] (t) -- (turnpt);
    \node[below] at (t) {$t$};
    
    \draw[dashed, gray] (Pt) -- (turnpt);
    \node[left] at (Pt) {$P(t)$};
    
    \fill[red] (turnpt) circle (3pt);
    \node[below right, red] at (turnpt) {$P'(t)=0$};
    
    \coordinate (v1) at (1.8, 2.0);
    \coordinate (v2) at (3.2, 2.0);
    \fill[orange] (v1) circle (2pt);
    \fill[orange] (v2) circle (2pt);
    
    \draw[orange, thick, <->] (v1) -- (v2) node[midway, above] {$P(r_1) = P(r_2)$};
    \draw[dashed, orange] (1.8, 0) -- (v1);
    \draw[dashed, orange] (3.2, 0) -- (v2);
    \node[below] at (1.8, 0) {$r_1$};
    \node[below] at (3.2, 0) {$r_2$};
    
\end{scope}

\begin{scope}[shift={(2,0)}]
    \draw[->, thick] (0,0) -- (5,0) node[right] {$v$};
    \draw[->, thick] (0,0) -- (0,4.5) node[above] {$f_V(v)$};
    
    \coordinate (vsing) at (0.8, 0);
    
    \draw[blue, very thick, domain=0.3:0.72, samples=50] 
    plot (\x, {1.2/sqrt(0.8-\x)});

    \draw[blue, very thick, domain=0.88:3.0, samples=50] 
    plot (\x, {1.2/sqrt(\x-0.8)});

    \fill[red] (vsing) circle (3pt);
    \node[below] at (vsing) {$P(t)$};
    
    \draw[dashed, red, thick] (vsing) -- (0.8, 4.5);
    
    \node[red, right] at (1.5, 3.0) {$f_V(v) \to \infty$};
    \node[red, right] at (1.5, 2.5) {as $v \to P(t)$};
    
\end{scope}

\end{tikzpicture}}
\caption{A turning point in the power function $P(r)$ correspond to a 
singularity in the probability density. \textbf{Left:} The power function has a minimum at $r = t$ 
where $P'(t) = 0$. Two distinct locations $r_1$ and $r_2$ produce the same power value (orange points). 
\textbf{Right:} When the transmitter location is randomized, the resulting probability 
density $f_V(v)$ has a singularity at $v = P(t)$ where the derivative $P'(r)$ vanishes. The singularity 
arises from the change-of-variables formula given by expression~\eqref{e.densitymonotonic}.}
\label{fig:singularity_mechanism}
\end{figure}
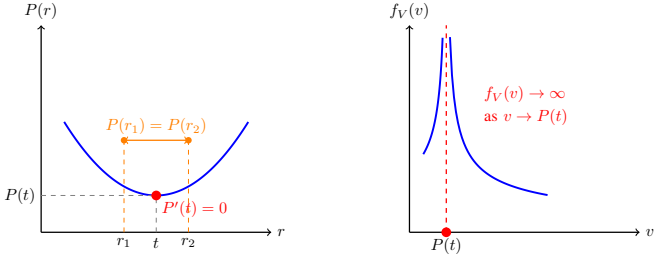

\subsubsection{A Single Turning Point} 
Turning points in the signal power function $P$ result in divisions by zero, creating singularities in the probability density of the randomized signal power, as Figure~\ref{fig:singularity_mechanism} illustrates. To examine this effect closer, we widen our sample interval to $(t_{i-1},t_{i+1})$ so that it covers one turning point of function $P$ located at $t_i$. We assume the function $P$ is twice differentiable in the neighborhood of the turning point $t_i$. For ease of exposition, we further assume the second derivative at the turning point is positive, so $P''(t_i) > 0$. (This assumption only simplifies the exposition and the mathematical argument without changing the overall message.) We define $h_i$ as the restriction of $P$ on the interval $(t_{i-1},t_{i+1})$. To give more explicit asymptotic expressions, we assume that the transmitter is randomly placed according to a uniform variable, but this assumption can be easily lifted. 

Now we state what happens asymptotically to the distribution and density of a randomized version of the signal power $P$.
\begin{Proposition}\label{prop:turningpoint}
Let $U$ be a continuous uniform random variable defined on the interval $(t_{i-1},t_{i+1})$ with probability distribution $F_U(u)=\P(U\leq u)$. Assume the above conditions on the function $h_i(r)$, which has one turning point on the interval $(t_{i-1},t_{i+1})$. The random variable $V_i=h_i(U)$ has a probability distribution $F_{V_i}(v)$, where $v\geq 0$, which behaves asymptotically such that 
\begin{align}\label{e.distsingular}
F_{V_i}(v)\sim\frac{2}{t_{i+1}-t_{i-1}} \left(\frac{2}{P''(t_i)}\right)^{1/2}
 \left(v-P(t_i)\right)^{1/2} \,,
\end{align}
as  $v\to P(t_i)$, where $v$ approaches $P(t_i)$ from above.
Furthermore, the probability density $f_{V_i}(v)$, where $v> 0$, behaves asymptotically such that
\begin{align}\label{e.densitysingular}
f_{V_i}(v)\sim\frac{1}{t_{i+1}-t_{i-1}}\left(\frac{2}{P''(t_i)}\right)^{1/2}\frac{1}{\left(v-P(t_i)\right)^{1/2}}\,,
\end{align}
as  $v\to P(t_i)$, where $P(t_i)\geq 0$ and $v\neq P(t_i)$. 
\end{Proposition}

\begin{proof}
See Appendix~\ref{s.proof.prop:turningpoint}.
\end{proof}
We see that, provided the second derivative is non-zero, $P''(t_i) \neq 0$, the singularity in the density $f_{V_i}$ is of the inverse (or reciprocal) square root kind. Although the probability density $f_{V_i}(v)$ is not defined at $v=P(t_i)$, it is an integrable singularity.
An obvious research inquiry is establishing which statistical methods deal with such singularities in this setting. 

\subsubsection{Multiple Turning Points}\label{sec:multiple_turning_points}
The previous analysis can be extended to multiple turning points of the signal power $P$. For a sufficiently well-behaved function $h$, we can partition its support into subintervals $A_1,\ldots,A_{\ell}$, resulting in strictly monotonic functions $g_1,\ldots,g_{\ell}$, where the support of each function $g_i$ is the subinterval $A_i$, as illustrated in Figure~\ref{fig:decomp}. (The partition includes an additional set $A_0$ containing all the endpoints of intervals, as well as other stray points where the function $h$ does not behave nicely.)   

This intuitive partitioning approach, however, seems to be rarely described in introductory textbooks or is mentioned in passing without proof~\cite[Section~2.7]{hogg2013mathematical},~\cite[Section~4.4, Remark~2]{rohatgi2015introduction},~\cite[Section~4.7, Equation~(16)]{grimmett2001probability}, which we discovered after extensive searching. Gut~\cite[Section~2.2]{gut2009intermediate} covers the result more thoroughly, but, for our work, the most relevant statement is given by Casella and Berger~\cite[Theorem~2.1.8]{casella2024statistical}. 

Skipping the finer mathematical details, the result basically says that for a random variable $X$ with probability density $f_X$, we obtain the probability density of a random variable $Y=h(X)$ by summing over the partition, giving
\begin{align}\label{e.densitymanytoone}
f_Y(y)=\sum_{i=1}^{\ell} f_X(g_i^{-1}(y))\left|\frac{d}{dy}g_i^{-1}(y)\right|\, & \text{ if }y\in\mathcal{Y}\,,
\end{align}
otherwise $f_Y(y)=0$, where $\mathcal{Y}$ is the support of the random variable $Y$. 

The above expression tells us that the number of singularities can potentially grow with the number of turning points. But the story is a little more subtle, as multiplicity is possible.  A turning point located at $t$ tells us that a singularity will occur, provided $t$ falls inside the sample window. But it is the power value $P(t)$ that tells us precisely where that singularity occurs. And naturally for a simple oscillating function,  different turning point locations $t_1,...,t_{\ell}$ can give the same power values, so $P(t_1)=\dots=P(t_{\ell})$ despite $t_1\neq\dots\neq t_{\ell}$, thus collapsing multiple singularities into one. The example of $h(r)=\sin(r)$, illustrated in Figure~\ref{fig:decomp}, would only generate a maximum of two singularities, even if the sample window was widened to include more turning points; also see the results of a simulated example in Figure~\ref{PPDFY2A}. In other words, more turning points does not necessarily imply more singularities in the probability density.

We now are prepared to apply the above results to a model with an oscillating power signal arising from reflections off walls. 

\section{Geometrical Signal Model}\label{sec:Geometrical}
Consider wireless communication between a receiving node located at the origin $O$ and a transmitting node located at the point $z$ with polar coordinates $(r,\theta)$. We adopt the standard assumption that the transmitted electromagnetic field can be represented by a plane wave that is independent of time and linearly polarized to be parallel to the ground. Following this representation, the transmitted signal at distance $r$ from the transmitter can be expressed as $E(r)=e^{jkr}$, where $j:=\sqrt{-1}$ and $k$ denotes the wave number, which is defined as $k:=2\pi/\lambda$ with $\lambda$ being the signal's wavelength. Thus $k$ is proportional to signal frequency $f$ via $k=2\pi f/c$ where $c$ is the speed of light. Assuming $c=3e8$ m/s,  a frequency range of $1$ to $100$ gigahertz corresponds to a wave number range of roughly $20.1$ to $2010$ rad/m. For our results, however, note that we chose the values of the wave number $k$ merely to illustrate concepts.

We will also incorporate attenuation into our model with a continuous function $\alpha(r)$ to denote the attenuation at distance $r$. 
A typical single-slope attenuation function is $\alpha(r)=r^{-\beta/2}$ with $\beta>0$ being the attenuation exponent, which is the model we will often assume. At distance $r$ from the transmitter, the signal is now $E(r)=\alpha(r)e^{jkr}$.

\section{One Wall}\label{sec:One Wall}
We assume the receiver $O$ and a single transmitter $z$ (with polar coordinates $(r,\theta)$) are located next to an infinitely long vertical wall passing through $x=a$; see Figure~\ref{fig:onewall}. In other words, the transmitter is located on the half-plane to the left of $x=a$, meaning the points $\{(x,y):x\leq a;x,y\in\R\}$, where $a\geq 0$ is the distance between the wall and the origin $O$.

When the signal reaches the wall, we assume that a fraction $\kappa\in[0,1]$ of the transmitted signal power is reflected by the wall, while the remaining power fraction $1-\kappa$ is absorbed. The superposition principle implies that the received signal at distance $r$ from $O$ is a deterministic quantity given by
\begin{align}
\hat{S}(r)&=\alpha(r)e^{jkr}+\sqrt{\kappa}\alpha(\hat{r}_1)e^{-jk\hat{r}_1}\\
&=\alpha(r)e^{jkr}-\sqrt{\kappa}\alpha(\hat{r}_1)e^{jk\hat{r}_1},\label{e.onewall}
\end{align}
where $\hat{r}_1$ is the distance traversed by the signal due to the single image. The first term on the right-hand side of~\eqref{e.onewall} is the line-of-sight signal. The second term is the non-line-of-sight signal for which we used the fact that each reflection adds $\pi$ to the signal phase, thus inverting the wave and explaining the negative sign.

The line-of-sight signal travels between $z$ and $O$, resulting in the horizontal and vertical distances $x:=r\cos\theta$ and $y:=r\sin\theta$, respectively. Then the distance from transmitter to the receiver via the wall is simply
\begin{align}\label{e.r1}
\hat{r}_1=\sqrt{(2a-x)^2+y^2}.
\end{align}

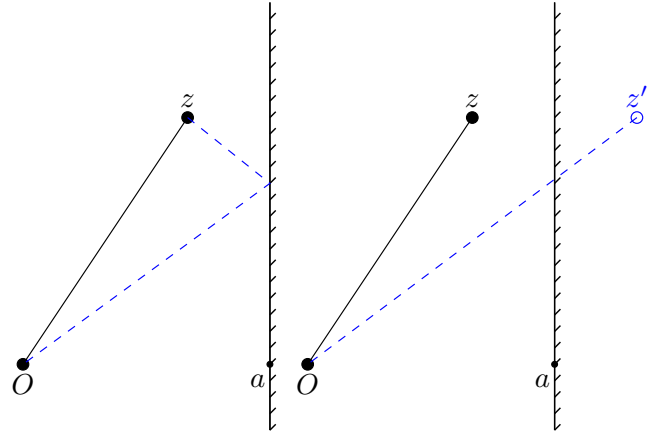
\begin{figure}[t]
\centering
\resizebox{\tikzmedium}{!}{\begin{tikzpicture}[scale=1]
\filldraw 
(2,0) circle (2pt) node[align=center, below] {$O$}       
(5,0) circle (1pt) node[align=right,  below] {\hspace{-1em}$\,a$} 
(4,3) circle (2pt) node[above] {$z$} ;  
\draw[-] (5,-.8) --(5,4.4); 
\draw[line width=.5pt,interface_right] (5,-.8) --(5,4.4); 
   \draw(2,0) -- (4,3)  ;  
    \draw[dashed][color=blue] (4,3) -- (5,2.2) 
        -- (2,0); 
\end{tikzpicture}
\begin{tikzpicture}[scale=1]
\filldraw 
(2,0) circle (2pt) node[align=center, below] {$O$}      
(5,0) circle (1pt) node[align=right,  below] {\hspace{-1em}$\,a$} 
(4,3) circle (2pt) node[above] {$z$} ;  
\draw[-] (5,-.8) --(5,4.4); 
\draw[line width=.5pt,interface_right] (5,-.8) --(5,4.4); 
   \draw(2,0) -- (4,3)  ;  
\draw[color=blue] (6,3) circle (2pt) node[above] {$z'$};  
    \draw[dashed][color=blue] (6,3) 
        -- (2,0); 
\end{tikzpicture}}
\caption{Two equivalent ways of representing a signal propagating between a transmitter and receiver, both located next to a reflecting wall. The black solid line is the line-of-sight signal, while the blue dashed line is the non-line-of-sight signal. In both representations the dashed blue lines are of equal total length. \textbf{Left:} The diagram illustrates the signal reflecting off the wall. \textbf{Right:} The diagram uses a virtual node or \emph{image} on the other side of the wall. }\label{fig:onewall}
\end{figure}

\begin{figure}[t]
\centering
\resizebox{\tikzmedium}{!}{\begin{tikzpicture}[scale=1]
    \filldraw 
(0,0) circle (1pt) node[align=right, below] {\hspace{1em}  -$b$} 
(2,0) circle (2pt) node[align=center, below] {$O$}     
(5,0) circle (1pt) node[align=right,  below] {\hspace{-1em}$\,a$} 
(4,3) circle (2pt) node[above] {$z$} ;  
   \draw[-] (0,-1) -- (0,5); 
   \draw[line width=.5pt,interface_left] (0,-.95) -- (0,4.95); 
  \draw[-] (5,-1) -- (5,5); 
  \draw[line width=.5pt,interface_right] (5,-.95) -- (5,4.95); 
   \draw(2,0) --(4,3)  ;  
\draw[color=blue] (6,3) circle (2pt) node[above] {$z'$} ;  
\draw[color=red] (-4,3) circle (2pt) node[above] {$z''$} ;  
    \draw[dashed][color=blue]  (6,3) --(2,0); 
    \draw[dashed][color=red] (-4,3)  --(2,0); 
\end{tikzpicture}}
\caption{A receiver $O$ and a transmitter $z$ placed between two infinitely long parallel walls. The black solid line represents the line-of-sight  signal path, while the red and blue dashed lines are the respective signals from the first left and right images. The distances between receiver $O$ and the images on the right and left are denoted by $\hat{r}_1$ and $\hat{\ell}_1$ respectively.}
\label{fig:twowalls}
\end{figure}
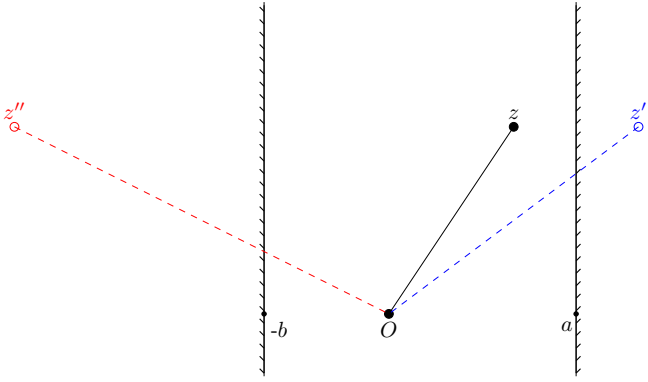



\section{Two Parallel Walls}\label{sec:Two_Walls}
We assume the receiver $O$ and a single transmitter $z$ are located between two infinitely long parallel walls; see Figure~\ref{fig:twowalls}. We again assume that a fraction $\kappa\in[0,1]$ of the transmitted signal power is reflected by the wall, while the remaining power fraction $1-\kappa$ is absorbed.

According to the method of images, as described in Section~\ref{sec:related.work}, the received signal at the receiver location $O$ is determined by an infinite number of images; also see Kyritsi's PhD thesis~\cite[Chapter~7]{kyritsi2001multiple} for details and references. In practice, however, the contributions from large-order images (or signal reflections) decay quickly due to signal absorption. The received signal at distance $r$ from $z$ is given by
\begin{align}\label{e.signalhat}
\hat{S}(r)=&\alpha(r)e^{jkr}\nonumber\\
&+\sum_{n=1}^{\infty}(-\sqrt{\kappa})^n\alpha(\hat{r}_n)e^{jk\hat{r}_n}\nonumber\\
&+\sum_{n=1}^{\infty}(-\sqrt{\kappa})^n\alpha(\hat{\ell}_n)e^{jk\hat{\ell}_n},
\end{align}
where $\hat{r}_n$ and $\hat{\ell}_n$ with $n=1,2,\ldots$ are the distances traversed by the signals due to the $n$-th right and left images, respectively. In the latter expression, again we used the fact that each reflection adds $\pi$ to the signal phase, thus inverting the wave, which explains the negative signs appearing in the signal reflection terms. For comparison purposes, we will often omit the line-of-sight term, thus expressing the signal only in terms of signal reflections, namely as
\begin{align}\label{e.signalnlos}
S(r)=\sum_{n=1}^{\infty}(-\sqrt{\kappa})^n\alpha(\hat{r}_n)e^{jk\hat{r}_n}+\sum_{n=1}^{\infty}(-\sqrt{\kappa})^n\alpha(\hat{\ell}_n)e^{jk\hat{\ell}_n}.
\end{align}

The signal's final ray coming from the left or right images gives two cases that depend on whether the ray traverses an even or an odd number of times between the two walls. For integer $n\geq 0$ and an even number of traversals, the distances of the reflected signals that come immediately from the right and left walls before arriving at $O$ are given, respectively, by
\begin{align}
\hat{r}_{n+1}&=\sqrt{(2n(a+b)+2a-x)^2+y^2},\label{e.rneven}\\
\hat{\ell}_{n+1}&=\sqrt{(2n(a+b)+2b+x)^2+y^2},\label{e.tneven}
\end{align}
where $a$ and $b$ denote the distances from $O$ to the right and left walls, respectively. For an odd number of traversals, the distances are obtained as
\begin{align}
\hat{r}_{n+1}&=\sqrt{((2n+1)(a+b)+a+b-x)^2+y^2},\label{e.rnodd}\\
\hat{\ell}_{n+1}&=\sqrt{((2n+1)(a+b)+a+b+x)^2+y^2}.\label{e.tnodd}
\end{align}

For the special symmetric case where $a=b$, and using the notation $d:=a+b$, equations~\eqref{e.rneven}--\eqref{e.tnodd} simplify to
\begin{align}
\hat{r}_{n+1}&=\sqrt{((n+1)d-x)^2+y^2},\label{e.rnsym}\\
\hat{\ell}_{n+1}&=\sqrt{((n+1)d+x)^2+y^2}.\label{e.tnsym}
\end{align}

\begin{Proposition}\label{prop:1}
For the special symmetrical case of $a=b=d/2$, the attenuation function $\alpha(r)=r^{-\beta/2}$, and transmitter at $z = (x,0)$ where $x$ is the horizontal distance from origin, the received non-line-of-sight signal at distance $r$ is given by the following closed-form expression:
\begin{align}\label{e.signallerch}
S(r)=&\frac{e^{-jkr}}{d^{\beta/2}}\left[\Phi\left(-\sqrt{\kappa}e^{jkd},\beta/2,-r/d\right)-(-d/r)^{\beta/2}\right]\nonumber\\
+&\frac{e^{+jkr}}{d^{\beta/2}}\left[\Phi\left(-\sqrt{\kappa}e^{jkd},\beta/2,+r/d\right)-(d/r)^{\beta/2}\right],
\end{align}
where $\Phi$ is the Lerch transcendent (function), defined in the DLMF~\cite[Section~25.14]{dlmf} by NIST as
\begin{align}\label{e.lerch}
\Phi(\zeta,s,\gamma):=\sum_{n=0}^{\infty}\frac{\zeta^n}{(n+\gamma)^s},
\end{align}
where $\gamma\neq 0,-1,\ldots$, $|\zeta|<1$, and $\Re\{s\}>1$.
\end{Proposition}

\begin{proof}
See Appendix~\ref{s.proof.prop:1}.
\end{proof}

\begin{Remark}
The convergence of the Lerch transcendent requires $\kappa < 1$, corresponding to the
physical requirement that some signal absorption occurs at each wall reflection.
The case $\kappa = 1$ (perfect reflection) lies outside the domain of convergence.

The Lerch transcendent often arises when summing fractions with power laws in the denominator. It is also known as the Hurwitz--Lerch zeta function, which is different from the Hurwitz or Lerch's zeta function. Despite such infinite sums appearing in models involving method of images, unexpectedly we did not find any relevant work that explicitly mentions this special function. 

The Lerch transcendent can be numerically evaluated in mathematical packages such as Maple or Mathematica. There is ongoing research on developing asymptotic methods and fast algorithms for evaluating these functions; see, for example, the recent paper by L{\'o}pez and Sinus{\'\i}a~\cite{lopez2025lerch}. For applications involving intelligent surfaces, efficient evaluation of the Lerch transcendent is important because real-time optimization algorithms may need to compute signal power across many spatial configurations when determining optimal phase settings.
\end{Remark}

\subsection{Received Signal Power}
In Figure~\ref{fig:powerx1}, we plot the (non-line-of-sight) signal power $P(x,y):=|S(x,y)|^2$ at the receiver location $(x,y)$ by setting $y=0$ and varying the $x$ values. We assume that the transmitting node location is fixed at the point $O$ with $a=b=0.5$, while the attenuation exponent is $\beta=4$. Recall that expressions $x=r\cos\theta$ and $y=r\sin\theta$ provide the Cartesian coordinates. As shown in Figure~\ref{fig:powerx1}, $P(x,0)$ oscillates, which contrasts with classical free-space propagation with no reflections, where the signal power decreases with increasing distance. The power oscillations are due to the existence of constructive and destructive wave interference, with the results being reminiscent of electromagnetic waves propagating through waveguides. Figure~\ref{fig:powerx2} demonstrates that the form of the power oscillations changes when $k$ increases to 100. The wave number $k$ is proportional to the signal frequency; hence, as $k$ increases, the waves oscillate more frequently. Similarly, when we fix $x=0$ for the receiver location and vary its vertical placement $y$, we see in Figure~\ref{fig:powery1} that the signal power $P(0,y)$ oscillates with varying $y$.

\begin{figure}[t]
\centering
\resizebox{\tikzmedium}{!}{\includegraphics{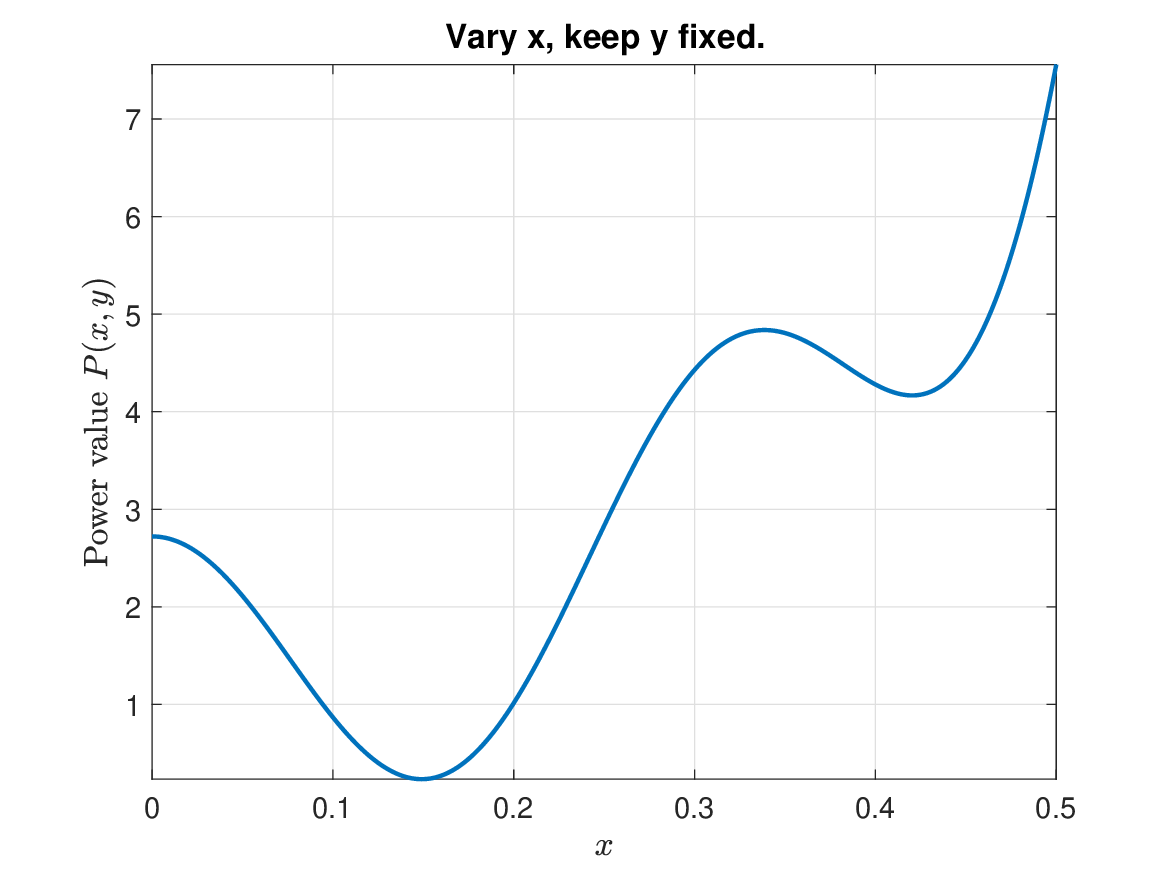}}
\caption{The (non-line-of-sight) signal power $P(x,0)$ for transmitter placement at $a=b=0.5$, attenuation exponent $\beta=4$, power reflection coefficient $\kappa=0.5$, wave number $k=10$, and different values of the horizontal receiver placement $x$.}\label{fig:powerx1}
\end{figure}

\begin{figure}[t]
\centering
\resizebox{\tikzmedium}{!}{\includegraphics{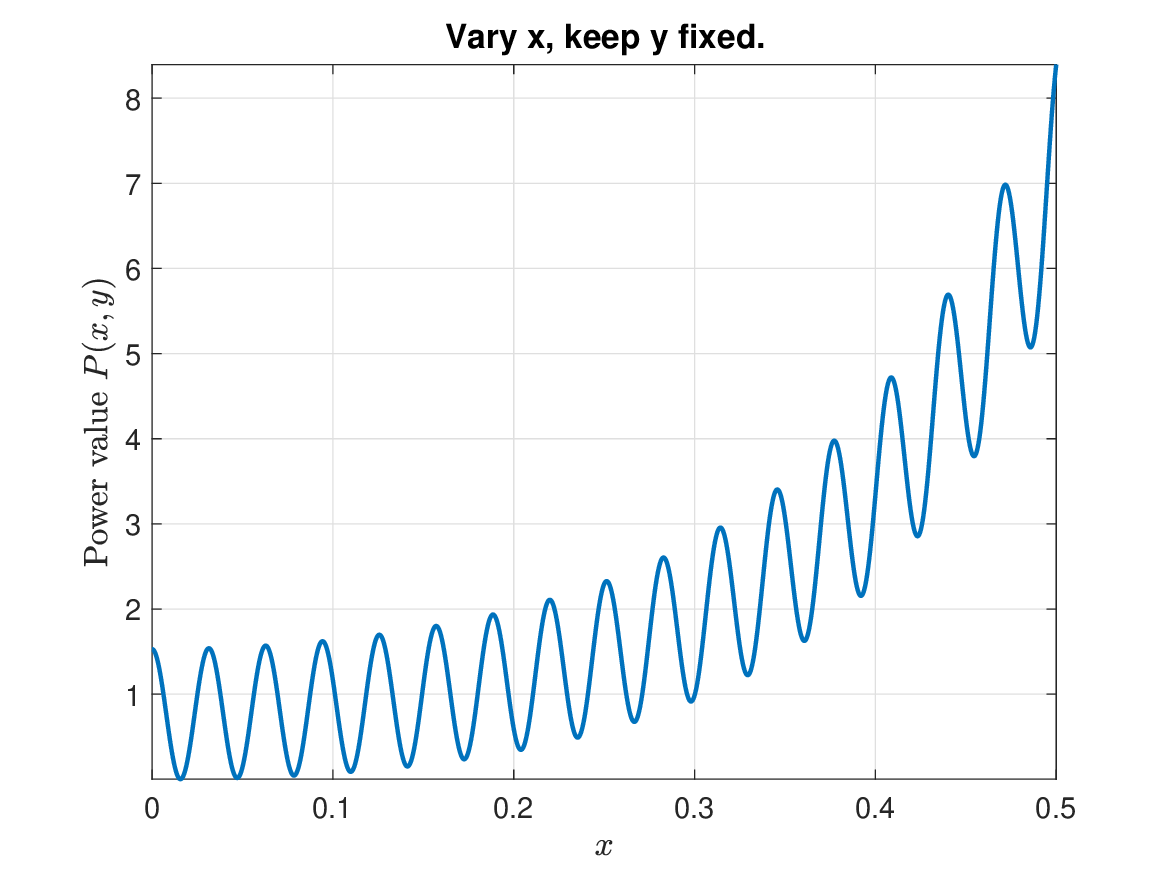}}
\caption{The signal power $P(x,0)$ for wave number $k=100$ and different values of the horizontal receiver placement $x$. The remaining parameters are the same as Figure~\ref{fig:powerx1}.}\label{fig:powerx2}
\end{figure}

\begin{figure}[t]
\centering
\resizebox{\tikzmedium}{!}{\includegraphics{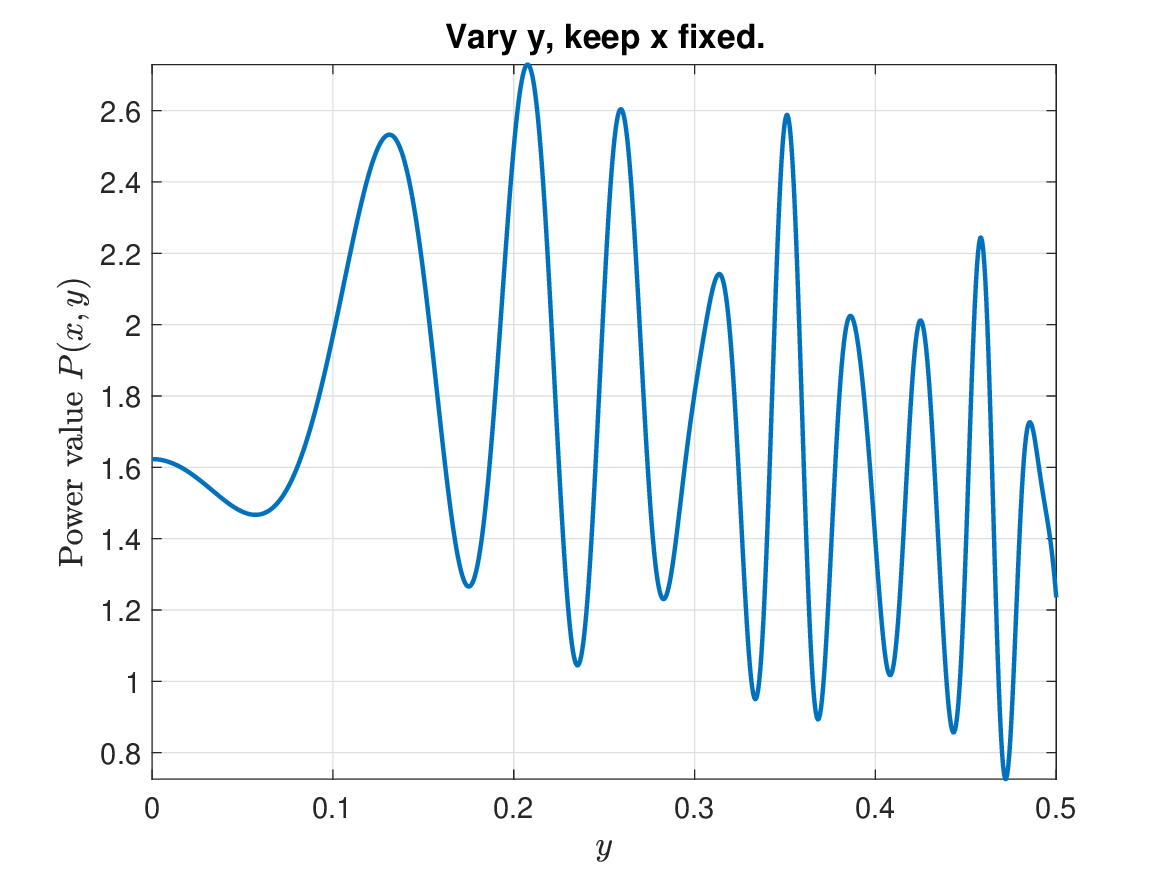}}
\caption{The signal power $P(0,y)$ for wave number $k=1000$ and different values of the horizontal receiver placement $y$. The remaining parameters are the same as Figure~\ref{fig:powerx1} and \ref{fig:powerx2}}\label{fig:powery1}
\end{figure}

\subsection{Intelligent Walls}
We now consider the case where the walls are coated with an intelligent surface. 
We assume an intelligent surface can adjust the phase of reflected signals by controlling the electromagnetic properties of the reflecting surface. Specifically, these idealized surfaces can:
\begin{itemize}
\item Remove the $\pi$ phase shift that normally occurs upon reflection
\item Add additional phase shifts to align reflected signals constructively
\item Dynamically tune the total phase to maximize received signal strength
\end{itemize}

In the optimal case, intelligent surfaces align all signal components---line-of-sight and all reflections---to interfere constructively with a common phase. If intelligent surfaces eliminate the $\pi$ phase shift from each reflection, the $(-\kappa)^n$ terms become $(\kappa)^n$ terms in the signal expression~\eqref{e.signalhat}.
However, the phases $kr$, $k\hat{r}_n$, and $k\hat{\ell}_n$ generally differ due to different path lengths, so the terms still do not align perfectly. In an idealized scenario where intelligent surfaces possess complete knowledge of all signal paths and can apply arbitrary path-dependent phase 
adjustments---a theoretical upper bound on performance rather than a 
practically achievable configuration---all signal components could be aligned 
to a common phase $e^{j \theta_0}$, yielding
\begin{align}\label{e.intelligentbound}
\hat{S}_0(r) := &\,e^{j\theta_0} \nonumber \\
 &\times \left[ \alpha(r) + \sum_{n=1}^{\infty} (\sqrt{\kappa})^n \alpha(\hat{r}_n) + \sum_{n=1}^{\infty} (\sqrt{\kappa})^n \alpha(\hat{\ell}_n) \right]
\end{align}

This upper bound assumes capabilities beyond current intelligent surfaces technology, 
serving primarily as a theoretical performance benchmark. Noting this, taking the magnitude of both sides and squaring of the signal expression~\eqref{e.intelligentbound} gives an upper bound for the signal power $\hat{P}(r)$, namely
\begin{align}
\hat{P}_0(r) &=\left[ \alpha(r) + \sum_{n=1}^{\infty} (\sqrt{\kappa})^n \alpha(\hat{r}_n) + \sum_{n=1}^{\infty} (\sqrt{\kappa})^n \alpha(\hat{\ell}_n) \right]^2\\
&\geq \hat{P}(r) \,.
\end{align}
This upper bound represents the maximum achievable signal power when intelligent surfaces eliminate destructive interference completely, serving as a benchmark for evaluating the performance of practical intelligent surface implementations.  

\subsection{Implications for Intelligent Surfaces}
When both walls are equipped with intelligent surfaces, the phase control capabilities offer opportunities to reshape the signal power function $P(x,y)$. Understanding the turning points becomes critical because:
\begin{itemize}
\item The singularities in probability densities occur at predictable locations determined by geometry.
\item Optimization algorithms for intelligent surfaces must account for these singularities when deployed in corridors or urban canyons.
\item Phase adjustments can potentially smooth the power function, reducing or eliminating turning points.
\end{itemize}

\subsection{Model Limitations}
Our two-wall propagation model incorporates several simplifying physical assumptions that merit discussion. The model assumes specular (mirror-like) reflection with an angular-independent reflection coefficient $\kappa$, neglecting diffuse scattering and angle-dependent reflection that occur in practice. We employ a far-field approximation where the phase is approximately uniform across the wall surface, which holds when distances are large compared to wall separation ($r \gg d$). The model treats electromagnetic waves as scalar quantities, omitting polarization effects and their interaction with wall materials. Additionally, we assume infinite parallel walls and a single-slope attenuation model with constant exponent $\beta$.

Despite these idealizations, the fundamental phenomenon we establish---that random spatial sampling of oscillatory signal power creates probability density singularities at turning points---is robust to these approximations. The singularities arise from the geometric relationship between transmitter location and multipath interference, not from the specific details of wave propagation. More sophisticated models incorporating near-field effects, angular-dependent reflection, or polarization would modify the detailed form of the power function $P(x,y)$ but would not eliminate the turning points or their associated singularities, provided constructive and destructive interference remains present. This robustness makes our results applicable to practical wireless systems where the idealized assumptions are only approximately satisfied.

\section{Random Geometrical Signal Models}\label{sec:Random_Model}
We consider natural ways to introduce randomness into our deterministic model. We examine two randomization approaches: random spatial positioning of the transmitter (Section~\ref{sec:Random_Model}A), which captures geometric variations, and random phase variation (Section~\ref{sec:Random_Model}B), which provides a more mathematically tractable alternative. We start with the first approach, which we already started discussing  in Section~\ref{sec:Turning}. 


\subsection{Random location model}
A natural way to introduce randomness into our wave-based model is to randomly position the transmitter, which can be justified by considering placements, movements and other uncertainties of transmitters. In Cartesian coordinates, the transmitter is located at $(x,y)=(r\cos\theta,r\sin\theta)$. If the transmitter is perturbed by a random vector $\mathbf{N}=(N_x,N_y)$, where the vector components are two random variables $N_x$ and $N_y$ symmetrical around zero (so they have zero mean) with variance $\sigma^2$, then the transmitter is now located at $(X',Y')=(x+N_x,y+N_y)$.

The random variables $N_x$ and $N_y$ may be, for example, two independent uniform or normal random variables with variance $\sigma^2$, where a natural choice for $\sigma$ would be a quantity proportional to the signal wavelength $\lambda$. In practice, we found it easier to randomly locate the transmitter using uniform random variables, which have bounded support. Otherwise, it is possible that the random perturbations are too large, sending the transmitter to the wrong side of the wall.

Under the random location model, the random perturbations will of course alter the distances of the rays, which we write as $\hat{R}_n$ and $\hat{L}_n$. In the symmetrical case when $a=b$, these random variables are simply
\begin{align}
\hat{R}_n&=\sqrt{(nd-x+N_x)^2+(y+N_y)^2},\label{e.R}\\
\hat{L}_n&=\sqrt{(nd+x+N_x)^2+(y+N_y)^2}.\label{e.L}
\end{align}

\subsubsection{Randomly Sampling the Signal} 
To better understand the random location model, we examine the situation when we just randomly vary $x$ in a uniform manner and keep $y$ fixed. This is equivalent to sampling the non-line-of-sight signal function $S(X,y)$, given by equation~\eqref{e.signallerch}, for a randomly uniform $X$. We can then estimate the density of the (probability) distribution $\P(P(X,y)\leq s)$ by performing an ordinary empirical count. (Based on our experiences here, standard kernel-based density estimation methods seem to fail at capturing the subtleties of the empirical probability density; see Remark~\ref{rem:density_est} below for details.) For each histogram, we took about $10^5$ samples with $200$ bins.  

\subsubsection{Peaks in the Probability Density} 
We now apply the mathematical framework developed in Section~\ref{sec:Turning} to our specific propagation model. As established there, turning points in the signal power function $P$ create singularities in the probability density under random spatial sampling. These peaks can be predicted from the power expression $P(x,y)=|S(x,y)|^2$. 

For fixed $y=0$ and varying $x\in(x_m,x_{i+1})$ on an interval where $P(x,0)$ is strictly monotonic, we can apply Proposition~\ref{prop:monotonic}. The constructive and destructive wave interference in our two-wall model creates turning points where $P'(x,0)=0$, which generate the singularities. Equation~\eqref{e.densitymonotonic} reveals that significant peaks can occur in the probability density even near (but not at) turning points where the derivative $P'(x)$ approaches zero.


\begin{Remark}
For environments with intelligent surfaces, the ability to predict turning point locations from geometry enables targeted phase optimization. Rather than treating fading as purely random, controllers for intelligent surfaces can leverage knowledge of these geometric singularities to improve signal reliability in specific spatial regions.
\end{Remark}

Constructive and destructive interference creates the turning points in the (non-random) signal power, which then give peaks in the probability density of the (random) signal power.

\subsubsection{Results}\label{sec:results.loc} 
Our results support our mathematical argument about singularities, with the values of the turning points corresponding to the peaks in the empirical density estimates. More specifically, for the first (or top) plot in Figure~\ref{PPDFX2}, we have again plotted $P(x,y)$ as we vary $x$, as we did for Figure~\ref{fig:powerx2}, but for a smaller $x$-domain. For the second (or bottom) plot in Figure~\ref{PPDFX2}, we have sampled uniformly $X$ on the interval $(0.15,0.35)$ and empirically estimated the probability density of the power term $P(X,y)$.

The power values where peaks in the probability density occur can be predicted, as they correspond to the turning points of the power term $P(x,y)$. We indicated the turning points and singularities respectively with red dots and red dashed lines in the two plots presented in Figure~\ref{PPDFX2}. We see in Figure~\ref{PPDFY3} that similar peaks appear when we randomly sample $Y$, keep $x$ fixed, and plot the empirical density of $P(x,Y)$.

If we either decrease the wave number $k$ or the interval or region of the transmitter location, then there will be fewer turning points and, consequently, fewer peaks in the probability density. Furthermore, if two (or more) turning points take the same power value, then they will appear as one peak in the empirical density, as we explained previously in Section~\ref{sec:multiple_turning_points}. Under our model, this happens if we fix $x$ and vary $y$ across a suitable interval, as there is symmetry in $y$ around $y=0$. The results of such an example are illustrated in Figure~\ref{PPDFY2A}, where we see there are three turning points in the power value, but only two peaks in the empirical density, as two turning points coincide in power value. 

On the other hand, even if the random location sampling does not cover a turning point in the power, its effects can still appear in the empirical probability density. We see this in Figure~\ref{PPDFY2B}, which we produced with the same parameters as Figure~\ref{PPDFY2A} but sampled on a smaller interval. The two turning points visible in Figure~\ref{PPDFY2B} produced two peaks in the empirical probability density, but there is a third peak (on the far right), which stems from the turning point in the power at $y=0$, which falls outside of the sample interval $(0.1,0.6)$.

Finally, we randomly translate in both Cartesian directions. We see in Figure~\ref{PPDFXY1} that the peaks still appear, but they are less significant, when we randomly sample both $X$ and $Y$. Note that it is somewhat difficult to obtain similar peaks in the empirical probability density, as the $x$ and $y$ variables work on different scales.

\begin{figure}[t]
\centering
\resizebox{\tikzlarge}{!}{\includegraphics{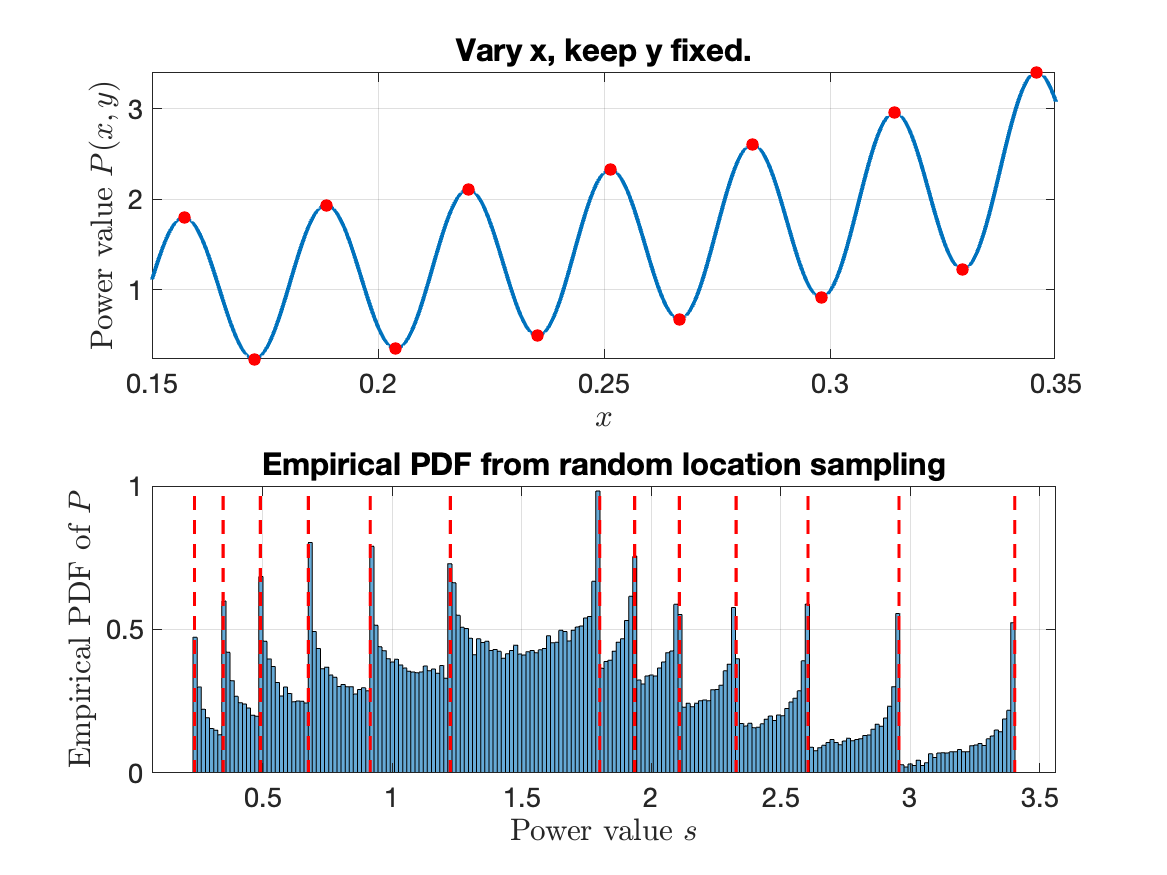}}
\caption{\textbf{Top:} The signal power $P(x,y)$ as $x$ varies. \textbf{Bottom:} Empirical probability density of $P(X,y)$ with varying $x\in(0.15,0.35)$ and fixed $y=0$. Thirteen turning points in the signal power produce thirteen peaks in its empirical probability density. Parameters $a=b=0.5$, $\beta=4$, $\kappa=0.5$, and $k=10^2$.}\label{PPDFX2}
\end{figure}

\begin{figure}[t]
\centering
\resizebox{\tikzlarge}{!}{\includegraphics{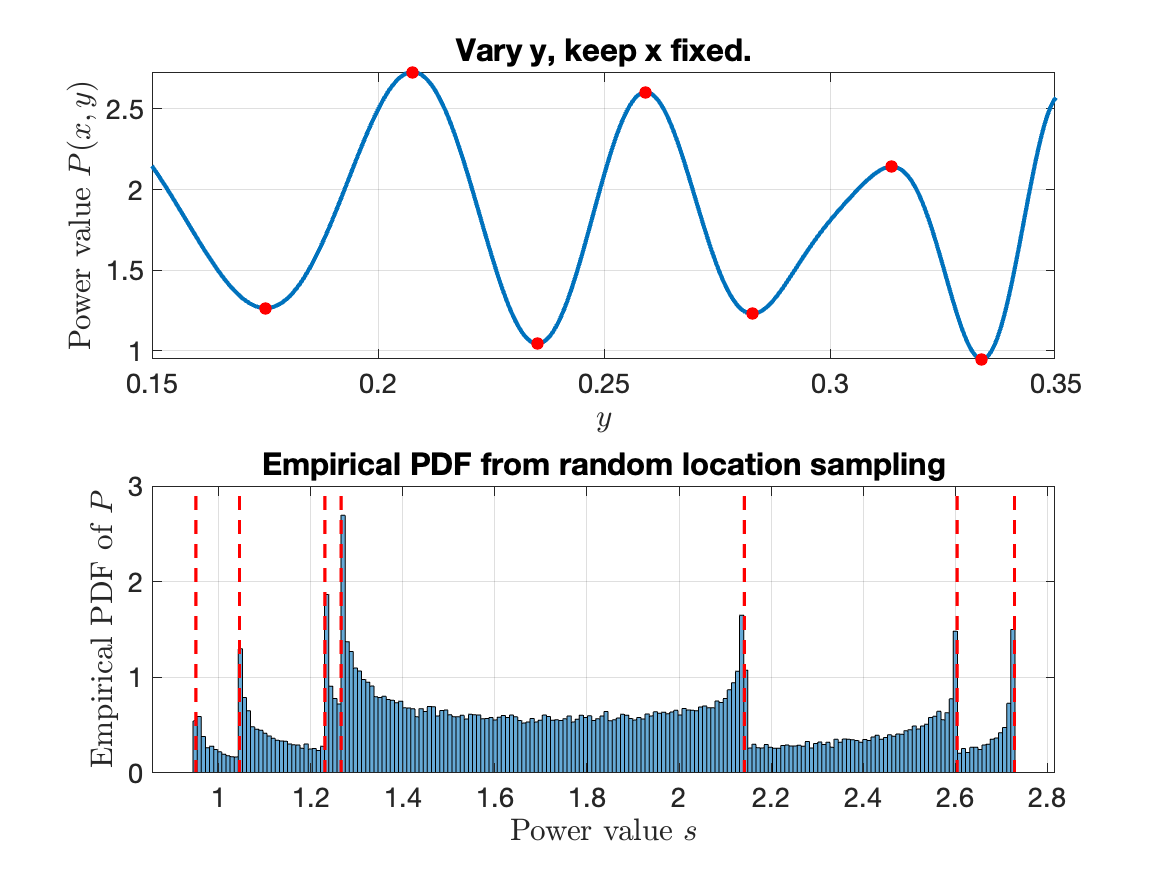}}
\caption{\textbf{Top:} The signal power $P(x,y)$ as $y$ varies.  \textbf{Bottom:} Empirical probability density of $P(x,Y)$ with varying $y\in(0.15,0.35)$ and fixed $x=0$. 
Parameters $a=b=0.5$, $\beta=4$, $\kappa=0.5$, and $k=10^3$.}\label{PPDFY3}
\end{figure} 

\begin{figure}[t]
\centering
\resizebox{\tikzlarge}{!}{\includegraphics{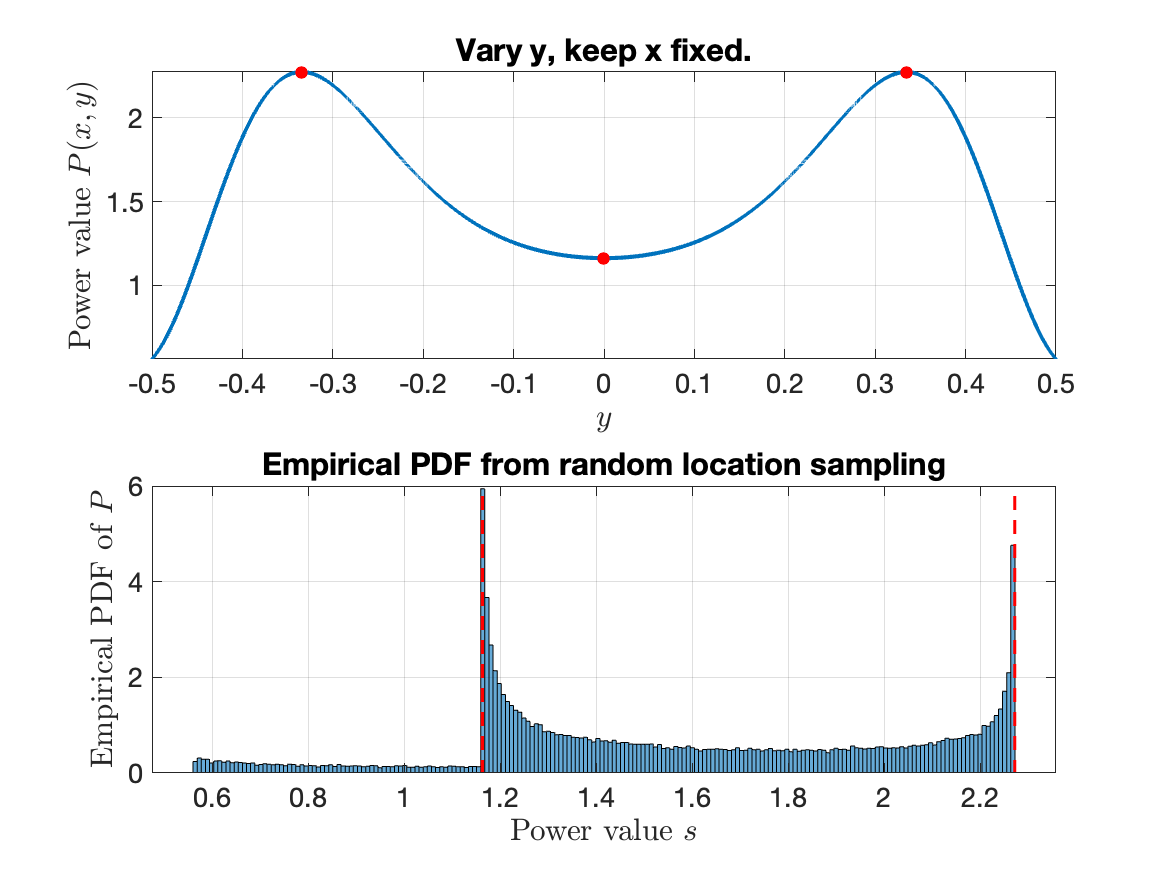}}
\caption{\textbf{Top:} The signal power $P(x,y)$ as $y$ varies. \textbf{Bottom:} Empirical probability density of $P(x,Y)$ with varying $y\in(-0.5,0.5)$ and fixed $x=0.1$. 
The three turning points in the power create only two peaks in its empirical probability density, because two of the turning points have the same power value, collapsing into a single peak. 
Parameters $a=b=0.5$, $\beta=4$, $\kappa=0.5$, and $k=10^2$.}\label{PPDFY2A}
\end{figure} 

\begin{figure}[t]
\centering
\resizebox{\tikzlarge}{!}{\includegraphics{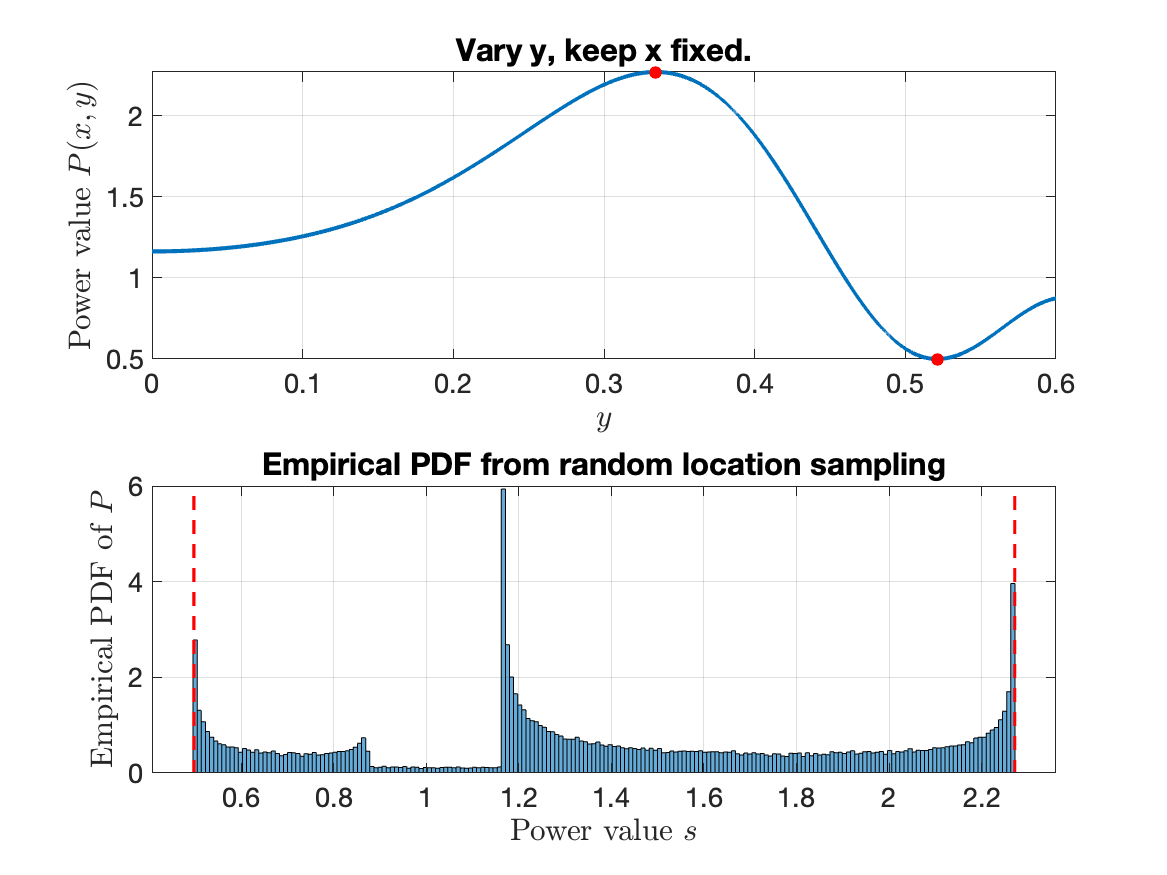}}
\caption{\textbf{Top:} The signal power $P(x,y)$ as $y$ varies. Only two true turning points of signal power $P(x,y)$ are visible, but a third turning point almost appears in the signal power on the far left with a power value of slightly below $1.2$.  \textbf{Bottom:} Empirical probability density of $P(x,Y)$ with varying $y\in(0,0.6)$ and fixed $x=0.1$.  There are three peaks in the empirical probability density. But the third peak (in the middle) is not a true singularity. It appears because the sampling nears the truncated third turning point of the signal power $P(x,y)$. Parameters $a=b=0.5$, $\beta=4$, $\kappa=0.5$, and $k=10^2$.}\label{PPDFY2B}
\end{figure}

\begin{figure}[t]
\centering
\resizebox{\tikzlarge}{!}{\includegraphics{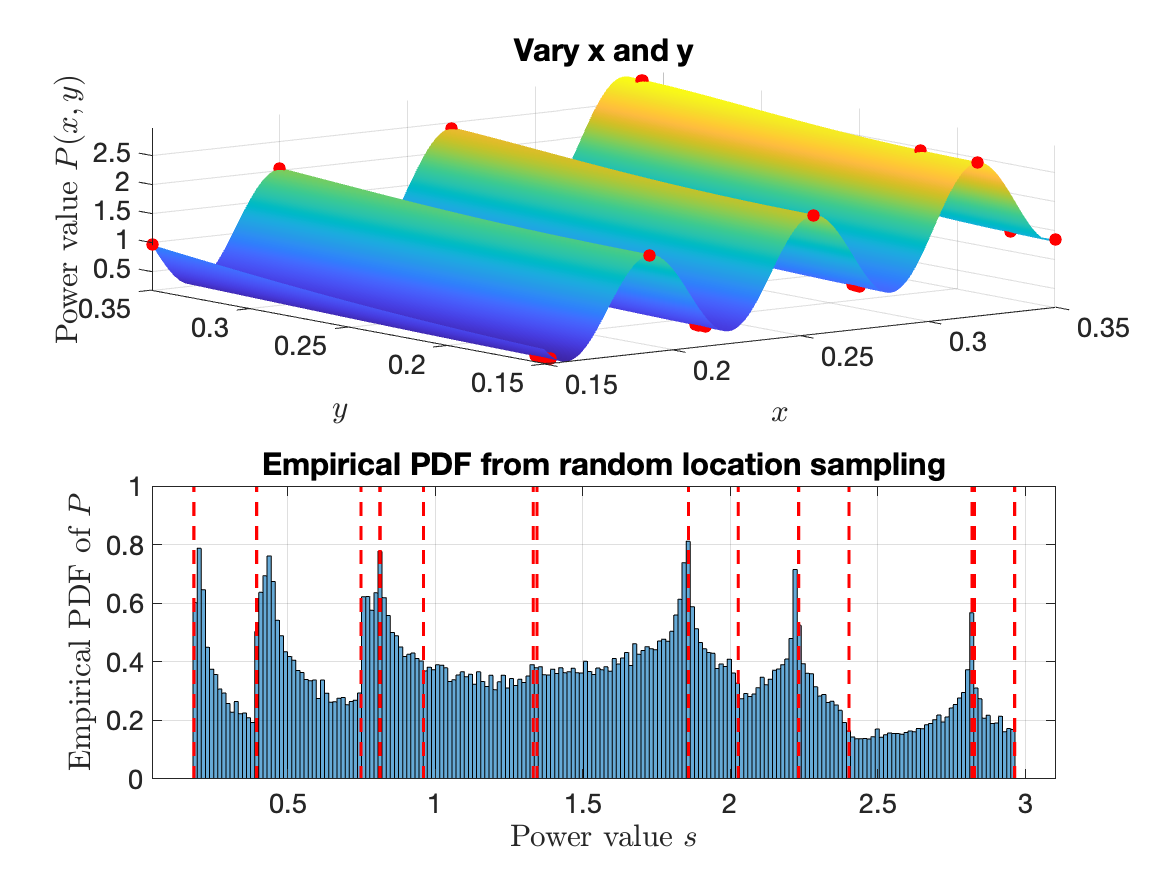}}
\caption{\textbf{Top:} The power of the signal term $P(x,y)$ as $x$ and $y$ both vary. \textbf{Bottom:} Empirical probability density of $P(X,Y)$ with varying $x\in(0.05,0.45)$ and varying $y\in(-0.2,0.2)$. Parameters $a=b=0.5$, $\beta=4$, $\kappa=0.5$, and $k=10$.}\label{PPDFXY1}
\end{figure} 

\begin{Remark}[Challenges for density estimation]\label{rem:density_est}
The singularities in our probability densities create fundamental challenges for standard statistical density estimation methods. Classical kernel density estimators assume smooth, differentiable densities, and the  asymptotic theory breaks down at discontinuities~\cite{silverman1986density,wand1994kernel,scott2015multivariate}. The convergence rates typically quoted for kernel method do not apply near singularities where the density is unbounded~\cite[Section 3.2.7.2 and 6.2.3.5]{scott2015multivariate}.


For densities with discontinuities or unbounded values, alternative approaches exist in the statistical literature, including methods based on transformation~\cite{wand1991transformations} or boundary kernels~\cite{muller1991smooth}. However, these typically address boundary problems or jump discontinuities, not the inverse square-root singularities arising in our model. Our histogram-based empirical approach avoids these pitfalls, though it introduces its own bias-variance trade-offs through bin width selection.
\end{Remark}

\subsection{Random Phase Model}
We want to propose another model, which is mathematically more tractable while still capturing the induced randomness of the previous model. For such a random model, we consider the phase terms in the two sums from the right and left set of images. We now assume that the phases of the right and left images, namely $k\hat{r}_n$ and $k\hat{\ell}_n$, (modulo $2\pi$) behave like coupled pairs of random variables $(\hat{U}_n,\hat{V}_n)$, which means the non-line-of-sight term becomes
\begin{align}
S(r)&=\sum_{n=1}^{\infty}\frac{(-\sqrt{\kappa})^n}{\hat{r}_n^{\beta/2}}e^{j\hat{U}_n}+\sum_{n=1}^{\infty}\frac{(-\sqrt{\kappa})^n}{\hat{\ell}_n^{\beta/2}}e^{j\hat{V}_n}.
\end{align}

We want to know if it is possible that this random phase model can mimic the random translation model above, which has the (random) phases $k\hat{R}_n$ and $k\hat{L}_n$, where we recall $\hat{R}_n$ and $\hat{L}_n$ are given by expressions~\eqref{e.R} and~\eqref{e.L}.

\subsubsection{Uniform Phases}\label{sec:uniform_phases}
A simple model is to assume that all the random pairs $(\hat{U}_n,\hat{V}_n)$ are equal in distribution to some single pair $(\hat{U},\hat{V})$. We then set $\hat{U}$ and $\hat{V}$ to be uniformly distributed variables on the interval $(0,2\pi)$.

There are two arguments for having uniform phases. First, under the random location model, if the variance of the random perturbation is large enough so that the resulting phases $k\hat{R}_n$ and $k\hat{L}_n$ are sufficiently large (greater than $2\pi$) and small (less than $0$), then the phases (modulo $2\pi$) will converge in distribution to that of the uniform variable on $(0,2\pi)$. For example, assume under the random location model that a random phase value $k\hat{R}_n$ is equal in distribution to a normal variable $Z$ with mean $\pi$ and variance $\sigma^2$; then
\[
Z\bmod(2\pi)\rightarrow\hat{U}\text{ as }\sigma\rightarrow\infty,
\]
which gives us the first reason for using uniform phases in our model.

The second reason is more technical, interpreting  the model as a discrete-time deterministic dynamical system. Each set of reflections of the $n$-th image represents a time step in this system. Then the phases of the right and left images, $k\hat{r}_n$ and $k\hat{\ell}_n$, form a dynamical system. The uniform distribution is the natural choice for such dynamical systems with the modulus operator, due to the trigonometric functions, evoking the \emph{equidistribution theorem}; see, for example, the book by Hasselblatt and Katok~\cite[Section~4.1.4]{hasselblatt2003first}. In less formal language, as the dynamical system evolves with more and more reflections occurring, the two phases $k\hat{r}_n$ and $k\hat{\ell}_n$ (modulo $2\pi$) are two points moving on the circle, eventually exploring the entire circle in a uniform manner.

\subsubsection{Results}\label{sec:results.phase}
We now present numerical results that validate our theoretical findings regarding random fading models and signal power distributions.

As explained in Section~\ref{sec:uniform_phases}, when random perturbations are sufficiently large in the random location model, the resulting random phases of all right and left images behave approximately uniformly on $[0,2\pi]$. However, for each pair of left and right images, the corresponding random phases remain dependent on each other. In other words, for any pair of images, the right and left phases marginally behave like uniform variables, but the entire signal structure maintains internal dependencies.

The marginal distributions of the phases in the random location model closely resemble uniform distributions, but they are not independent. This correlation structure distinguishes the random location model from simpler models that assume complete independence. In contrast, under the far-field approximation model, the random phases can be treated as independent, which significantly simplifies analysis but may not capture the full complexity of the physical propagation environment.

Our numerical simulations confirm that the peaks in the probability density of signal power, as predicted by our analysis in Section~\ref{sec:Random_Model}, align precisely with the turning points of the deterministic signal power function. This validates our theoretical framework connecting wave interference patterns to statistical properties of random fading.

We compare the empirical probability densities under the random location and random phase models in Figures~\ref{PhaseLocPDFX2} and~\ref{PhaseLocPDFX3}. For $k=10^2$ (Figure~\ref{PhaseLocPDFX2}), the random location model exhibits multiple sharp peaks corresponding to turning points in the deterministic power function, as predicted by our theoretical analysis. In contrast, the random phase model produces a smooth, unimodal density without these singularities. The random phase model's density is approximately symmetric and resembles classical fading distributions like Rayleigh, lacking the geometric structure that creates peaks in the random location model.

At higher frequencies with $k=10^3$ (Figure~\ref{PhaseLocPDFX3}), the distinction becomes more pronounced. The random location model maintains its characteristic peaks at turning point values, though the increased frequency creates more oscillations in the underlying power function. The random phase model continues to produce a smooth density, demonstrating that uniform phase randomization fundamentally differs from spatial randomization. This comparison confirms that the singularities in probability densities arise specifically from the geometric coupling between random spatial sampling and deterministic signal oscillations, not merely from phase randomness.

\begin{figure}[t]
\centering
\resizebox{\tikzmedium}{!}{\includegraphics{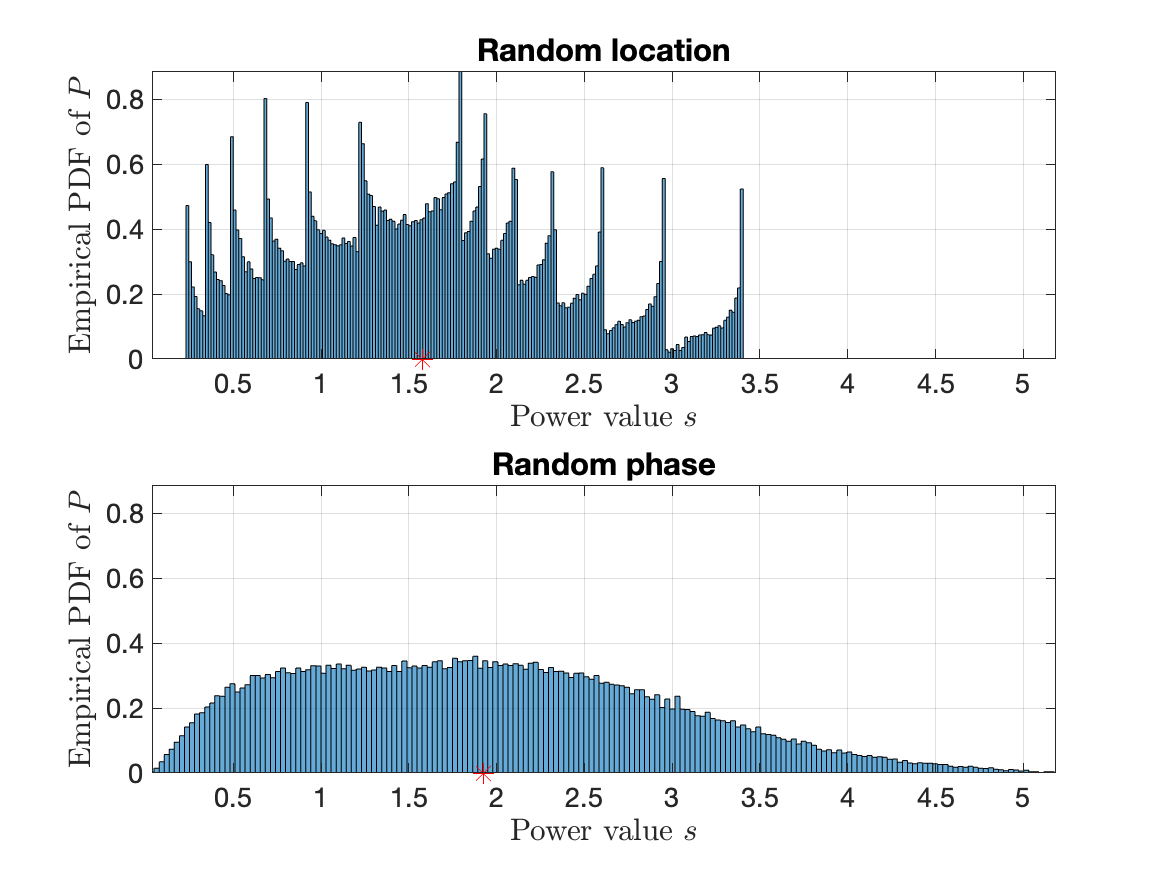}}
\caption{At lower wave number values, the two models give rather different results.
\textbf{Top:} Empirical probability density of $P(x,Y)$ under random location model: Varying $x\in(0.15,0.35)$ and fixed $y=0$. 
\textbf{Bottom:} Empirical probability density of $P(x,Y)$ under random phase model: Fixed $x=0.25$ and $y=0$. 
Parameters $a=b=0.5$, $\beta=4$, $\kappa=0.5$, and $k=10^2$.}\label{PhaseLocPDFX2}
\end{figure}

\begin{figure}[t]
\centering
\resizebox{\tikzmedium}{!}{\includegraphics{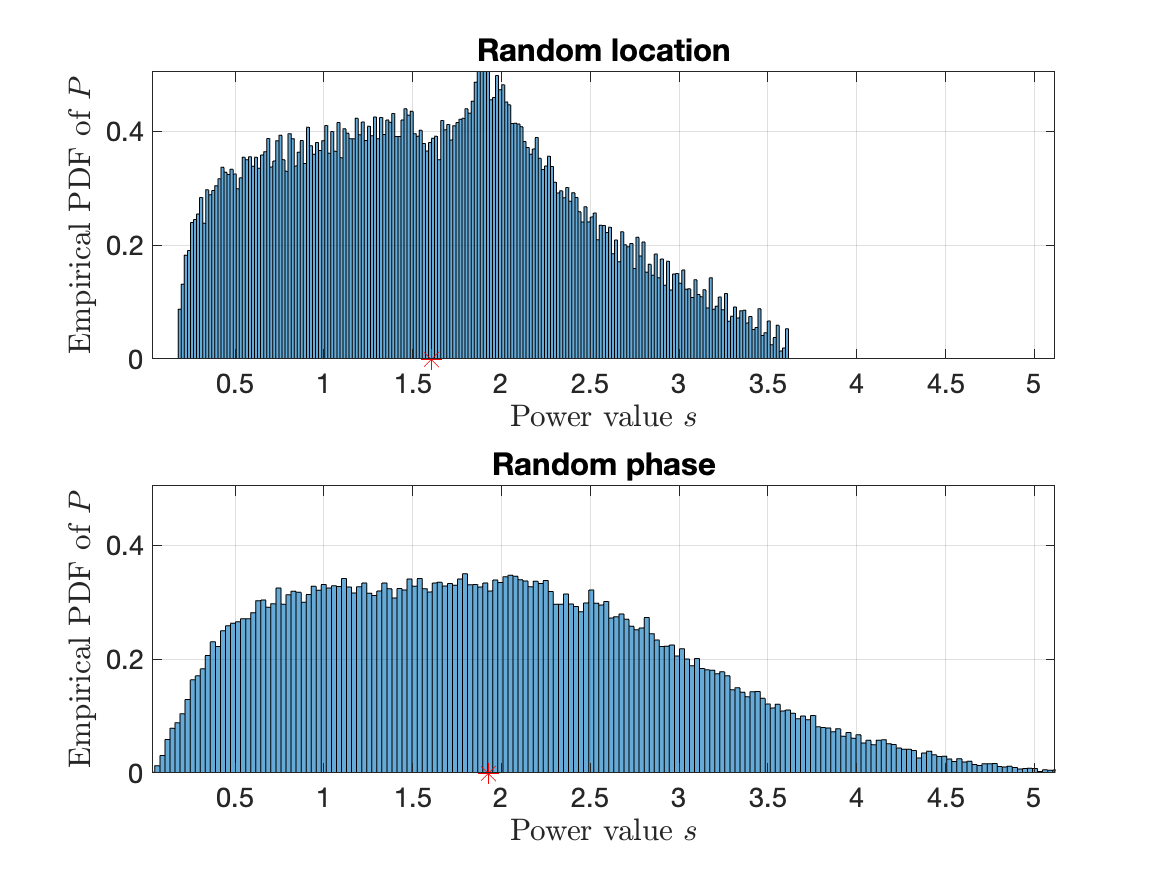}}
\caption{At higher wave number values, the two models give similar results.
\textbf{Top:} Empirical probability density of $P(x,Y)$ under random location model: Varying $x\in(0.15,0.35)$ and fixed $y=0$. 
\textbf{Bottom:} Empirical probability density of $P(x,Y)$ under random phase model: Fixed $x=0.25$ and $y=0$. 
Parameters $a=b=0.5$, $\beta=4$, $\kappa=0.5$, and $k=10^3$.}\label{PhaseLocPDFX3}
\end{figure}


\section{Conclusion}\label{sec:Conclusion}
Using probability fundamentals, we showed that randomly moving a transmitter can result in peaks or rather singularities in the probability density of the (randomized) received signal power, provided there are turning points in the distance-based propagation model. This has direct implications for how to model and statistically study wireless signals. To apply these findings, we have developed a physical-statistical model for signal propagation between parallel walls using the method of images. Our approach bridges the gap between purely deterministic ray-tracing methods and purely statistical fading models by incorporating geometric structure directly into the stochastic framework.

For the symmetric two-wall configuration, we derived a closed-form expression for the received signal using the Lerch transcendent function, which reveals the oscillatory nature of signal power due to constructive and destructive interference. This analytical result provides insight into waveguide-like behavior in cities that standard free-space propagation models cannot capture.

Under the random location model, we demonstrated both theoretically and numerically that oscillations in signal power generate singularities in the probability density function. These singularities occur precisely at the turning points of the deterministic signal power function, which can be predicted from the geometry of the configuration. For low wave numbers or frequencies, these peaks become particularly prominent, representing a regime where simplified random phase models fail to capture the underlying physical phenomena.

Our work establishes a foundation for understanding signal fading in environments with intelligent surfaces. The predictability of probability density singularities from geometric configuration provides a structural principle for optimization of intelligent surfaces: rather than treating all spatial locations equivalently, network operators can account for the known locations of turning points and resulting singularities. The geometric approach we have developed can accommodate phase-tunable walls and provides a framework for analyzing how such surfaces can mitigate destructive interference while accounting for the statistical properties induced by random user positioning. Future work can extend these results to more complex geometries and incorporate the full capabilities of intelligent surfaces into the stochastic geometry framework.

\section*{Data Availability}
The MATLAB code used to generate all numerical results and figures in this paper is available online~\cite{keeler_walls_fading}.

\appendices
\section{Proof of Proposition~\ref{prop:monotonic}}\label{s.proof.prop:monotonic}
Standard probability techniques gives us the probability distribution of the random variable $V_i=g_i(U)$ precisely because $g_j$ is a one-to-one function. (The use of this crucial assumption on the function is typically taught in probability courses, but later we will need to consider many-to-one functions.) We then derive the probability density by using the chain rule and a general formula for the derivative of an inverse function. If a continuous function $f$ is monotonic on the interval $(a,b)$, then the derivative of its inverse $f^{-1}$ is given by 
\begin{align}\label{e.inversederiv}
\frac{df^{-1}(v)}{dv}=\frac{1}{f'(f^{-1}(v))}, \quad v\in(a,b)\,,
\end{align}
which is closely related to a result sometimes called the inverse function theorem; see, for example, Tao~\cite[Section~10.4]{tao2016analysis}. Formula~\eqref{e.inversederiv} is also presented in classic texts on analysis such as Spivak~\cite[page~34]{spivak1965calculus} or Rudin~\cite[page~114]{rudin1964principles}.

\section{Proof of Proposition~\ref{prop:turningpoint}}\label{s.proof.prop:turningpoint}
Recalling $P'(t_i)=0$ and taking the Taylor expansion of $P(r)$ at the turning point $t_i$, we define the function
\begin{align}
\tilde{h}_i(r):=P(t_i)+\frac{1}{2}P''(t_i)(r-t_i)^2,
\end{align}
for which $\tilde{h}_i(r)\to h_i(t_i)=P(t_i)$ as $r\to t_i$. The function $\tilde{h}_i(r)$ is of course not strictly monotonic, but we can suitably partition its support (the subset of its domain where it is positive), giving two strictly monotonic inverses, which we write as a single expression
\begin{align}\label{e.inversesquare}
r=\tilde{h}_{m,\pm}^{-1}(v):=t_i\pm[2(v-P(t_i))/P''(t_i)]^{1/2}.
\end{align}
The rest of the proof follows from the standard steps of finding the distribution of the square of a random variable; see, for example, Gut~\cite[Section~2.2]{gut2009intermediate} or Hogg, McKean, and Craig~\cite[Section~1.7]{hogg2013mathematical}. More specifically, if $X$ is a continuous random variable with distribution $F_X$, then the distribution of the random variable $Y=X^2$ is $F_Y(y)=F_X(\sqrt{y})-F_X(-\sqrt{y})$,
while the probability density is $f_Y(y)=(f_X(\sqrt{y})-f_X(-\sqrt{y}))/(2\sqrt{y})$
. The same reasoning gives the probability distribution
\begin{align}
\Prob(\tilde{h}_i(U)&\leq v)\nonumber\\ 
=\,&\Prob(U\leq t_i+[2(v-P(t_i))/P''(t_i)]^{1/2})\nonumber\\
&-\Prob(U\leq t_i-[2(v-P(t_i))/P''(t_i)]^{1/2}).
\end{align}
Assuming $U$ is uniformly distributed on $(t_{i-1},t_{i+1})$ with distribution $\Prob(U\leq u)= F_U(u)=(u-t_{i-1})/(t_{i+1}-t_{i-1})$, then direct substitution into the above expression gives the first asymptotic formula~\eqref{e.distsingular}, and differentiation of that formula yields the second formula~\eqref{e.densitysingular}. Note: If we wanted to another distribution for the random variable $U$, we just replace the distribution $F_U$ with another distribution, and then derive its Taylor expansion. 

\section{Proof of Proposition~\ref{prop:1}}\label{s.proof.prop:1}
Based on our geometrical model, the reflected rays travel distances for $n\geq 1$, which are given by $\hat{r}_n=nd-r$ and $\hat{\ell}_n=nd+r$ for the $n$-th right and left images, respectively. For this case, the series expression~\eqref{e.signalnlos} for $S(r)$ simplifies to
\begin{align}\label{eq:A1}
S(r)=&\sum_{n=1}^{\infty}\frac{(-\sqrt{\kappa})^n}{(nd-r)^{\beta/2}}e^{jk(nd-r)}\nonumber\\
+&\sum_{n=1}^{\infty}\frac{(-\sqrt{\kappa})^n}{(nd+r)^{\beta/2}}e^{jk(nd+r)}.
\end{align}
We can write the first the infinite sum as
\begin{align}
\nonumber\sum_{n=1}^{\infty}&\frac{(-\sqrt{\kappa})^n}{(nd-r)^{\beta/2}}e^{jk(nd-r)}\\
\nonumber&=\frac{e^{-jkr}}{d^{\beta/2}}\sum_{n\ge 1}\frac{(-\sqrt{\kappa}e^{jkd})^n}{(n-r/d)^{\beta/2}}\\
\nonumber&=\frac{e^{-jkr}}{d^{\beta/2}}\left(\sum_{n=0}^{\infty}\frac{(-\sqrt{\kappa}e^{jkd})^n}{(n-r/d)^{\beta/2}}-(-d/r)^{\beta/2}\right)\\
&=\frac{e^{-jkr}}{d^{\beta/2}}\left(\Phi\left(-\sqrt{\kappa}e^{jkd},\beta/2,-r/d\right)-(-d/r)^{\beta/2}\right).\label{eq:A2}
\end{align}
Substituting expression~\eqref{eq:A2} into expression~\eqref{eq:A1} and performing the same derivation steps for the second infinite sum in expression ~\eqref{eq:A1} yields the final expression~\eqref{e.signallerch}.

\bibliographystyle{IEEEtran}
\bibliography{../shared/bibliography/walls,IEEEabrv}
\end{document}